\documentclass[11pt]{article}

\usepackage{amsmath,amssymb,amsthm}
\usepackage{natbib}
\usepackage{hyperref}
\usepackage{booktabs}

\usepackage{graphicx}
\usepackage{subcaption}
\usepackage{overpic}

\usepackage[margin=1in]{geometry}

\newtheorem{assumption}{Assumption}
\newtheorem{theorem}{Theorem}
\newtheorem{remark}{Remark}
\newtheorem{lemma}{Lemma}

\title{Kernel-Based Functional Balancing for Causal Inference with Compositional Treatments}

\author{
Sungbum Kim \qquad Jiayi Wang\thanks{Corresponding author.} \\
\vspace{0.15cm}
\textit{Department of Mathematical Sciences, University of Texas at Dallas}
}

\date{June 2026}

\begin{document}

\maketitle

\begin{abstract}
We study causal effect estimation with compositional treatments, where the exposure lies on a simplex and the estimand is defined over compositions rather than scalar or binary values. By considering a projection of the average potential outcome onto the treatment space, a kernel-based covariate functional balancing approach is adopted for weight construction. The weights are obtained by directly minimizing a worst-case balancing error over a reproducing kernel Hilbert space (RKHS) defined on the joint space of treatments and covariates, instead of being estimated under a treatment assignment model.
Building on these weights, an augmented weighted estimator (AWE) is proposed, where the outcome function is estimated via kernel ridge regression and combined with a marginal augmentation over the covariate distribution. Despite the complex structure of the resulting objective, a finite-dimensional convex optimization problem is formulated via a representer theorem and a low-rank approximation. The proposed estimator achieves $\sqrt{n}$-consistency without requiring consistent estimation or smoothness of the weights. An asymptotic normality result is established around a sample-specific target. Empirical performance is demonstrated through simulation studies and a real data application.
\end{abstract}

\vspace{0.5cm}

%
%
\section{Introduction} \label{s:intro}

Many interventions in science and policy involve allocating a finite resource rather than a single action.
Such interventions are naturally represented as compositional variables, where the components of a vector are proportions that sum to one. 
These data structures arise across disciplines and are increasingly common in empirical research.
In finance, a portfolio is a compositional treatment that allocates capital across assets subject to a fixed budget constraint \citep{Boonen.etal:2018, VegaGamez.AlonsoGonzalez:2024}. 
In the health and behavioral sciences, time-use epidemiology studies daily activity patterns---such as sleep, sedentary behavior, and physical activity---that sum to 24 hours \citep{Dumuid.etal:2020, Carson.etal:2016, Verswijveren.etal:2022}. 
Similar structures appear in molecular biology, where compositional data analysis compares biomarker profiles in glycomics \citep{Bennett.etal:2025}, and in marketing, where budget allocations and market shares form compositional systems \citep{Morais.etal:2018, FerrerRosell.etal:2021}.
In these settings, the objective is to quantify the causal effect of a compositional intervention on an outcome.
Under the potential outcomes framework, this is characterized by the expected outcome given a specific treatment allocation.
We focus on the average treatment effect (ATE), defined as the expected outcome under a specific compositional treatment.

Causal inference with compositional treatments is complicated by the simplex constraint, under which component proportions must sum to one.
A common approach is to map compositional variables to an unconstrained Euclidean space using log-ratio transformations or related representations \citep{aitchison1982statistical, egozcue2003isometric}, and then apply standard methods.
However, these transformations hinder causal interpretation.
Effects defined in the transformed space do not correspond to reallocations in the original composition, making it difficult to relate estimated effects to feasible interventions.
This limitation motivates approaches that operate directly on the simplex.
In this work, causal effects are defined through reallocations across components \citep{li2023relative}.

Existing approaches for estimating causal effects can be broadly categorized into regression-based methods and weighting methods.
Regression-based methods are commonly used to estimate the ATE due to their simplicity and interpretability.
These methods model the outcome as a function of treatment and covariates and use the fitted model to construct ATE estimators \citep{robins1995semiparametric, vanderlaan2011targeted}. 
They can be extended using nonparametric or machine learning techniques to accommodate complex outcome functions.
However, such extensions do not account for the reallocation structure inherent to compositional treatments and may yield estimates that are difficult to interpret.
Moreover, naive regression-based plug-in estimators inherit the first-stage estimation error. When flexible nonparametric or machine learning methods are used, both regularization bias and overfitting prevent the naive plug-in estimator from achieving $\sqrt{n}$-consistency \citep{chernozhukov2018double}.

An alternative approach is based on weighting.
Weighting methods construct a pseudo-population in which covariates are independent of treatment assignment.
Inverse probability weighting (IPW) assigns weights based on the inverse of the treatment assignment probability and can be interpreted as a density ratio that shifts the sample toward a target population \citep{rosenbaum1983central}.
However, IPW can produce unstable estimators when assignment probabilities are close to zero, leading to large or highly variable weights.
Moreover, its performance depends on accurate estimation of the treatment assignment mechanism, which becomes challenging in high-dimensional or constrained settings.
To address these issues, covariate balancing methods have been proposed.
Rather than modeling the treatment assignment mechanism, these methods construct weights that explicitly balance covariate distributions across treatment levels.
This perspective can improve numerical stability and reduce sensitivity to model misspecification. 
Examples include entropy balancing \citep{hainmueller2012entropy}, covariate balancing propensity scores \citep{imai2014covariate}, and kernel balancing \citep{hazlett2020kernel}.
Recent work has extended this perspective beyond binary treatments to continuous and more general treatment settings.
In particular, balancing-based approaches have been developed for continuous treatments \citep{fong2018covariate, tubbicke2022entropy, huling2024independence}, further extended to multivariate continuous treatments \citep{chen2023causal}, and even to functional treatments where the treatment is an infinite-dimensional object \citep{wang2023flexible}.
Despite these developments, causal inference with compositional treatments, where interventions correspond to reallocations on the simplex, remains largely underexplored.

In this paper, we propose a kernel-based balancing framework for causal effect estimation with compositional treatments. 
To capture complex, unstructured interactions between the treatment composition and covariates, we develop a general functional balancing approach. By assuming the outcome function lies in a tensor product RKHS defined on their joint domain, this model can adapt to highly nonlinear and interconnected confounding relationships. We construct weights to directly target the estimation error of the causal effect through a worst-case balancing criterion over the RKHS unit ball, without explicitly modeling the treatment assignment mechanism or requiring any smoothness assumption on the weight function. 
The resulting estimator achieves $\sqrt{n}$-consistency under regularity conditions. 
We also propose an augmented estimator that combines the balancing weights with an outcome regression fitted via kernel ridge regression. 
By doing so, the estimation error is bounded by the residual of the outcome regression rather than the RKHS norm of the outcome function itself. 
Since this residual vanishes as the sample size increases, this approach improves statistical efficiency.
Under regularity conditions, the augmented estimator is asymptotically normal with a tractable variance representation. 
Notably, these results are obtained without sample splitting or cross-fitting, in contrast to double machine learning approaches \citep{chernozhukov2018double, newey2018cross}.

The remainder of the paper is organized as follows. Section~\ref{s:preliminaries} presents the formal problem setup and key concepts. Section~\ref{s:Methodology} describes the proposed kernel-based balancing methodology and its computational aspects. Section~\ref{s:AWE} introduces the Augmented Weighted Estimator and presents its implementation using kernel ridge regression. Section~\ref{s:theory_section} establishes the theoretical properties of the estimator. Sections~\ref{s:simulation} and \ref{s:real_data} present simulation studies and a real-data application. 

%
%
\section{Preliminaries} \label{s:preliminaries}

\subsection{Causal Framework and Interpretable Modeling Objective}

We consider the problem of estimating the causal effect of a compositional treatment $A \in \mathcal{A}$ on a scalar outcome $Y$, where $\mathcal{A}$ denotes the $(L-1)$-dimensional simplex.
Let $X \in \mathcal{X}$ be a vector of covariates, where $\mathcal{X} \subseteq \mathbb{R}^p$ denotes the $p$-dimensional covariate space.
We adopt the potential outcomes framework, denoting $Y(a)$ as the potential outcome for a unit under treatment $a \in \mathcal{A}$. 
Identification of causal effect from observational data $\{(A_i, X_i, Y_i)\}_{i=1}^n$ relies on the following standard assumptions commonly adopted
in causal inference \citep{rosenbaum1983central, imbens2015causal}.
\begin{assumption}[Causal Identification Conditions]
    \label{ass:causal_identification}
    The following conditions are assumed: \\
    \textbf{(Consistency)} $Y = Y(A)$. \\
    \textbf{(Positivity)} For any treatment value $a \in \mathcal{A}$ with positive marginal density with respect to the simplex measure and any $x \in \mathcal{X}$, the conditional density satisfies $\varphi_{A|X}(a \mid x) > 0$. \\
    \textbf{(No Unmeasured Confounding)} $Y(a) \perp A \mid X$ for all $a \in \mathcal{A}$.
\end{assumption}
Under these conditions, the mean potential outcome $\mathbb{E}[Y(a)]$ is identified from the observed data as $\mathbb{E}_X \! \left[\mathbb{E}[Y \mid A = a, X]\right]$, where $\mathbb{E}_X[g(X)] := \int g(x)\varphi_X(x)\,dx$ denotes expectation with respect to the marginal distribution of the covariates $X$.

While the mean potential outcome may be complex, we summarize the causal effect through its best linear approximation by projecting $Y(A)$ onto the linear span of the treatment vector. Specifically, we define the target estimand $\boldsymbol{\beta}^*$ as
$$ 
    \boldsymbol{\beta}^* = \underset{\boldsymbol{\beta} \in \mathbb{R}^L }{\text{argmin}} \, 
    \mathbb{E} \Bigl[ \bigl( Y(A) - A^\intercal \boldsymbol{\beta} \bigr)^2 \Bigr],
$$
where the expectation is taken over the marginal distribution of $A$ and the distribution of the potential outcomes. 
Because the compositional treatment $A$ lies on the simplex, an intercept is omitted to avoid perfect collinearity and ensure identifiability. 

The minimizer admits a closed-form expression. Let $\boldsymbol{\Sigma}_A := \mathbb{E}[A A^\intercal]$. If $\boldsymbol{\Sigma}_A$ is nonsingular, then
\begin{equation}
    \boldsymbol{\beta}^* = \boldsymbol{\Sigma}_A^{-1} \mathbb{E}[A Y(A)].
    \label{eq:beta_star_def}
\end{equation}

The quantity $a^\intercal \boldsymbol{\beta}^*$ can be interpreted as the best linear approximation to the average potential outcome $\mathbb{E}[Y(a)]$.
This definition preserves the intuitive interpretability of the coefficients in $\boldsymbol{\beta}^*$, with each coefficient representing the marginal contribution of its corresponding component within the linear projection.

To give these coefficients a causal interpretation, we adopt the reallocational effects framework \citep{li2023relative}.
Under the unit-sum constraint, a change in treatment is represented as a reallocation of a fixed amount $\delta$ across components.
Specifically, consider reallocating $\delta$ from component $j$ to component $i$, so that $a' = a + \delta(e_i - e_j)$.
Based on this linear projection, the resulting change in the expected outcome is
\begin{equation}
    \tau (a', a) = \mathbb{E}[Y(a')] - \mathbb{E}[Y(a)] = (a' - a)^\intercal \boldsymbol{\beta}^* = \delta(\beta_i^* - \beta_j^*).
    \label{eq:tau} 
\end{equation}
More generally, an amount $\delta$ can be reallocated to component $i$ from the remaining $L-1$ components, yielding
$$
    \tau (\bar{a}', a) = \delta\left(\beta_i^* - \bar{\beta}_{-i}^*\right),
    \quad
    \bar{a}' = a + \delta \left( e_i - \frac{1}{L-1} \sum_{j \ne i} e_j \right),
    \quad
    \bar{\beta}_{-i}^* := \frac{1}{L-1}\sum_{j \neq i}\beta_j^*.
$$
For example, in a time-use study with $\delta = 1/24$, this corresponds to reallocating one hour to activity $i$ from the remaining activities.

\subsection{Oracle Weights for Confounding Adjustment}

Building on the previous section, our objective is to estimate $\boldsymbol{\beta}^*$ defined in \eqref{eq:beta_star_def}.
The primary challenge lies in estimating the term $\mathbb{E}\bigl[A \,  Y(A)\bigr]$ from observational data.
A naive sample average is generally biased in the presence of confounding, as the covariates $X$ may influence both the treatment allocation $A$ and the outcome $Y$.
To address this issue, a common strategy is to reweight the observed outcomes to construct a pseudo-population in which the treatment assignment can be viewed as independent of the covariates.

Theoretically, this reweighting is achieved through an oracle weight function
\begin{equation}
    w^*(a,x) = \frac{\varphi_A(a)}{\varphi_{A\mid X}(a \mid x)},
    \label{eq:w*definition}
\end{equation}
where $\varphi_A$ denotes the marginal density of the compositional treatment $A$, and $\varphi_{A\mid X}$ its conditional density given covariates $X$.
The role of $w^*(a, x)$ can be understood through its effect on expectations. For any integrable function $f(A,X)$,
$$
    \mathbb{E} \left[ w^*(A,X) f(A,X) \mid A=a \right] = \int f(a,x) \, \varphi_X(x) \, dx.
$$
In this sense, the weighted pseudo-population behaves as though treatment assignment were independent of the covariates.

However, directly estimating these oracle weights poses substantial practical and theoretical challenges.
The difficulty stems from the need to estimate both the marginal density $\varphi_A$ and the conditional density $\varphi_{A\mid X}$.
Nonparametric density estimation methods suffer from the curse of dimensionality, particularly when both $A$ and $X$ are multi-dimensional.
Parametric alternatives, such as Dirichlet regression, offer a seemingly convenient solution but are highly susceptible to model misspecification, which may result in residual confounding \citep{aitchison1982statistical}.
Moreover, the presence of zero or near-zero components renders the model ill-defined, further complicating estimation.
These challenges motivate approaches that bypass density estimation.
In particular, functional balancing methods \citep{wong2018kernel} achieve covariate balance without explicitly modeling the treatment assignment mechanism.
We adapt this idea to compositional treatments and develop the proposed method in the next section.

%
%
\section{Kernel-Based Covariate Functional Balancing Methodology} \label{s:Methodology}

\subsection{Problem Setup and Functional Representation}

Building on the observational setting described in Section~\ref{s:preliminaries}, our goal is to estimate the causal estimand $\boldsymbol{\beta}^*$. 
Suppose that the outcome satisfies the model
$$ Y_i = m(A_i, X_i) + \varepsilon_i, $$
where $m(A,X)$ is the true mean response function and $\mathbb{E}[ \varepsilon_i \mid A_i, X_i] = 0$.
Under Assumption~\ref{ass:causal_identification}, the estimand $\boldsymbol{\beta}^*$ admits the representation
$$
    \boldsymbol{\beta}^* = \boldsymbol{\Sigma}_A^{-1}\,\mathbb{E}\Bigl[A \cdot \mathbb{E}_X\bigl[m(A,X)\bigr]\Bigr].
$$

To connect the population representation with an implementable estimator, let $w_i^* := w^*(A_i, X_i)$ denote weights that satisfy the oracle balancing condition in \eqref{eq:w*definition}.
Let $w^*=(w_1^*, \dots, w_n^*)^\intercal$, define the diagonal weight matrix $\mathbf W^* = \mathrm{diag}(w_1^*, \dots, w_n^*)$, and denote $\mathbf A = (A_1, \dots, A_n)^\intercal$ and $\mathbf Y = (Y_1, \dots, Y_n)^\intercal$. The least squares estimator using the outcome weighted by the oracle weights is
$$
    \hat{\boldsymbol{\beta}}^* = (\mathbf A^\intercal \mathbf A)^{-1}\mathbf A^\intercal \mathbf W^* \mathbf Y.
$$
Under Assumption~\ref{ass:causal_identification}, the oracle weights satisfy
$$
    \mathbb{E} \left[ A \, w^*(A,X) \, Y \right] = \mathbb{E} \left[A \cdot \mathbb E_X[m(A,X)]\right] = \boldsymbol{\Sigma}_A \boldsymbol{\beta}^*,
$$
which implies $\mathbb E\bigl[\hat{\boldsymbol{\beta}}^* \bigr] = \boldsymbol{\beta}^*$. 
In practice, the oracle weights are unknown. We consider the following general weighted estimator with the weight $w = (w_1,\dots, w_n)^\intercal$
$$ 
    \hat{\boldsymbol{\beta}}(w) = (\mathbf A^\intercal \mathbf A)^{-1}\mathbf A^\intercal \mathbf W\mathbf Y,
$$
where $\mathbf W = \mathrm{diag}(w_1, \dots, w_n)$. In the following, we investigate the criteria needed to estimate the weights. 

We decompose the estimation error of an estimator $\hat{\boldsymbol{\beta}}(w)$. Using $\mathbf Y = \mathbf m + \boldsymbol{\varepsilon}$, where $\mathbf m = (m(A_1,X_1), \ldots, m(A_n,X_n))^\intercal$ and $\boldsymbol{\varepsilon} = (\varepsilon_1, \ldots, \varepsilon_n)^\intercal$, and applying the triangle inequality gives the decomposition
\begin{subequations} \label{eq:error_decomp}
\begin{align}
    \left\| \hat{\boldsymbol{\beta}}(w) - \boldsymbol{\beta}^* \right\|_2 \le& 
        \underbrace{\left\| (\mathbf{A}^\intercal \mathbf{A})^{-1} \mathbf{A}^\intercal \left( \mathbf{W} \mathbf{m} - \mathbb{E}_{X} [ \mathbf{m} ] \right) \right\|_2}_{\text{Stochastic Balancing Error}}
        \label{eq:error_decomp_a} \\
    & + 
        \underbrace{\left\| (\mathbf{A}^\intercal \mathbf{A})^{-1} \mathbf{A}^\intercal \mathbb{E}_{X} [ \mathbf{m} ] - \boldsymbol{\beta}^* \right\|_2}_{\text{Approximation Error}}
        \label{eq:error_decomp_b} \\
    & + 
        \underbrace{\left\| (\mathbf{A}^\intercal \mathbf{A})^{-1} \mathbf{A}^\intercal \mathbf{W} \boldsymbol{\varepsilon} \right\|_2}_{\text{Stochastic Noise Error}}
        \label{eq:error_decomp_c}
\end{align}
\end{subequations}
where $ \mathbb{E}_X[\mathbf{m}] := \left( \mathbb{E}_X \left[ m(A_1, X) \right], \dots,  \mathbb{E}_X \left[ m(A_n, X) \right] \right)^\intercal $. 
The first term represents the stochastic balancing error due to the discrepancy between the weighted sample $\mathbf{W}\mathbf m$ and their target marginal mean $\mathbb{E}_X[\mathbf m]$, which arises from finite-sample variability. 
The second term is a deterministic approximation error resulting from approximating the potentially nonlinear marginal mean function with a linear projection. 
The third term captures the variance induced by the weighted aggregation of the noise terms.

Our objective is to construct a weight vector $\hat w = (\hat{w}_1, \dots, \hat{w}_n)^\intercal$ that controls the stochastic balancing error in \eqref{eq:error_decomp_a} while regularizing the stochastic noise term in \eqref{eq:error_decomp_c}.
We denote by $\hat{\mathbf{W}} = \mathrm{diag}(\hat{w}_1, \dots, \hat{w}_n)$ the corresponding diagonal weighting matrix.
Since the projection term in \eqref{eq:error_decomp_b} does not depend on the weights, it is not affected by the choice of weights. Accordingly, the weight construction problem reduces to balancing the marginal mean function while stabilizing the variance induced by weighting. 
In practice, the true outcome function $m$ is unknown, so the balancing objective cannot be formulated directly in terms of $m$. To obtain a tractable formulation, we impose structural assumptions on the mean response function, specifically that $m$ belongs to the RKHS $\mathcal{H}$.

Inspired by \eqref{eq:error_decomp_a}, we introduce a general balancing criterion for an arbitrary function $f \in \mathcal{H}$. 
Let
$$ 
    \mathbf{f} = \bigl( f(A_1, X_1), \dots, f(A_n, X_n) \bigr)^\intercal, 
$$
and define the empirical marginal mean
$$ 
    \hat{\mathbb{E}}_X[\mathbf{f}] :=
        \begin{pmatrix}
            \frac1n \sum_{j=1}^n f(A_1, X_j) \\
            \vdots \\
            \frac1n \sum_{j=1}^n f(A_n, X_j)
        \end{pmatrix},
$$
which approximates $\mathbb{E}_X [f(A_i, X)]$. For the weight vector $w$, we define the balancing criterion function
\begin{equation}
    S(w, f) =
        \left\| (\mathbf{A}^\intercal \mathbf{A})^{-1} \mathbf{A}^\intercal \bigl( \mathbf{W} \mathbf{f} - \hat{\mathbb{E}}_{X}[\mathbf{f}] \bigr) \right\|_2^2.
    \label{eq:S(w, f)}
\end{equation}
To control imbalance uniformly over the function class, we minimize the worst-case discrepancy over the unit ball 
$\mathcal{H}(1):=\{ f \in \mathcal{H} : \|f\|_{\mathcal{H}} \le 1 \}$, 
leading to the objective
$$ \sup_{f \in \mathcal{H}(1)} S(w,f). $$

To control the stochastic noise term in \eqref{eq:error_decomp_c}, 
we introduce a variance regularization penalty of the form
\begin{equation}
    p(w) = \left\| (\mathbf{A}^\intercal \mathbf{A})^{-1} \mathbf{A}^\intercal \mathbf{W} \right\|_F^2.
    \label{eq:p(w)}
\end{equation}
This penalty controls variance inflation due to weighting.

We then construct the weight vector by solving the regularized worst-case balancing problem
\begin{equation}
    \hat{w} = \underset{w \in \mathbb{R}^n_+}{\text{argmin}} \, \left\{ \sup_{f \in \mathcal{H}(1)} S(w,f) + \lambda\, p(w) \right\}
    \label{eq:w_hat_def}
\end{equation}
where $\lambda > 0$ controls the trade-off between covariate balance and variance control. 
This optimization problem completes the formulation of the proposed kernel-based balancing estimator.

%
%

\subsection{Kernel-Based Balancing and Computational Implementation}

The optimization problem in \eqref{eq:w_hat_def} involves a supremum over the infinite-dimensional function space $\mathcal{H}(1)$, which presents analytical and computational challenges. 
The following theorem shows that the inner optimization of \eqref{eq:w_hat_def} admits a finite-dimensional representation.

\begin{theorem} [Representer and Decomposition Theorem]
\label{thm:representer}
    Let $\mathcal{H}$ be a reproducing kernel Hilbert space (RKHS) with product kernel
    $$ K\bigl((A,X),(A',X')\bigr) = K_A(A,A')\,K_X(X,X'). $$
    Let $\mathcal{H}(1)=\{f\in\mathcal{H}:\|f\|_{\mathcal{H}}\le1\}$ denote the unit ball of $\mathcal{H}$. 
    Then the worst-case balancing error satisfies
    $$ 
        \sup_{f\in\mathcal{H}(1)} S(w,f) = 
            \left\| (\mathbf{A}^\intercal \mathbf{A})^{-1} \mathbf{A}^\intercal 
            \bigl( \mathbf{W}, \, -\mathbf{I} \bigr)
            \tilde{\mathbf{K}}^{1/2} \right\|_{op}^2,
    $$
    where $\|\cdot\|_{op}$ denotes the operator norm, 
    $\mathbf{I}$ is the $n\times n$ identity matrix, 
    and $\bigl( \mathbf{W}, \, -\mathbf{I} \bigr)$ denotes the $n \times 2n$ block matrix formed by concatenating $\mathbf{W}$ and $-\mathbf{I}$.
    The matrix $\tilde{\mathbf{K}}$ is the $2n\times 2n$ kernel matrix defined in Appendix~\ref{app:formulas}; see \eqref{eq:k_tilde_mat} for its explicit form.
\end{theorem}

Directly forming and decomposing the full $2n \times 2n$ kernel matrix $\tilde{\mathbf{K}}$ can be computationally prohibitive for large $n$. To ensure scalability, we therefore employ a low-rank approximation of the kernel matrix using its rank-$r$ truncated eigenvalue decomposition, $\tilde{\mathbf{K}} \approx \mathbf{P}_1 \mathbf{Q}_1 \mathbf{P}_1^\intercal$, where $\mathbf{P}_1 \in \mathbb{R}^{2n \times r}$ contains the top $r$ eigenvectors and $\mathbf{Q}_1$ is a diagonal matrix of the corresponding eigenvalues.
$$
    \sup_{f \in \mathcal{H}(1)} S(w, f) \approx \|\mathbf{D}_w \|_{op}^2
    \quad \text{where} \quad 
    \mathbf{D}_w := (\mathbf{A}^\intercal \mathbf{A})^{-1}\mathbf{A}^\intercal 
        \bigl( \mathbf{W}, \, -\mathbf{I} \bigr)
        \mathbf{P}_1 \mathbf{Q}_1^{1/2}
$$

As a result, our original infinite-dimensional optimization problem reduces to the following finite-dimensional optimization problem
\begin{equation}
    \hat{w} \approx \underset{w \in \mathbb{R}^n_+}{\text{argmin}} \Bigl\{ \|\mathbf{D}_w \|_{op}^2 + \lambda p(w) \Bigr\}.
    \label{eq:w_hat_modified}
\end{equation}

The optimization problem in \eqref{eq:w_hat_modified} is efficiently solvable due to its convexity. 
As shown in Theorem~\ref{thm:convexity}, the objective is convex in $w$, ensuring global optimality and enabling efficient gradient-based optimization.

\begin{theorem}[Convexity]
\label{thm:convexity}
    Let $w \in \mathbb{R}^n$ and define $\mathbf{W} = \mathrm{diag}(w)$. 
    Let $\mathbf{P}_1$ be a fixed matrix, $\mathbf{Q}_1 \succeq 0$ be diagonal, and let $\lambda \ge 0$. 
    Define 
    $$ 
        \mathcal{L}(w) = \left\| \mathbf{D}_w \right\|_{op}^2 + \lambda p(w).
    $$
    Then $\mathcal{L}(w)$ is a convex function of $w$.
\end{theorem}
The convexity result facilitates efficient optimization via projected gradient descent subject to the non-negativity constraint ($w_i \ge 0$).
The gradient of $\mathcal{L} (w)$ with respect to $w$ is given by 
$$
    \frac{\partial}{\partial w} \mathcal{L}(w) = 
    2 \cdot \text{diag}\left( \mathbf{P}_1\mathbf{Q}_1^{1/2} v \right) \mathbf{A} (\mathbf{A}^\intercal \mathbf{A})^{-1} \mathbf{D}_w v
    +
    2 \lambda \cdot \text{diag}\left( \mathbf{A}(\mathbf{A}^\intercal \mathbf{A})^{-2}\mathbf{A}^\intercal \right) \odot w
$$
where $v$ denotes the leading right singular vector of $\mathbf{D}_w$, and $\odot$ denotes the element-wise product.

These analytical gradients enable the implementation of a projected gradient descent algorithm to compute $\hat{w}$ under the non-negativity constraint. 
The regularization parameter $\lambda$ controls the trade-off between the balancing objective and the stability of the resulting weights. 
However, selecting $\lambda$ remains a practical challenge, since standard tuning procedures such as cross-validation are designed to optimize predictive performance and do not directly target covariate balance or causal estimation \citep{hastie2009elements}. 
While our theoretical results suggest an asymptotically optimal choice with $\lambda \asymp 1$, we propose a practical, data-driven strategy for tuning $\lambda$ in Section~\ref{s:lambda}.

%
%
\section{The Augmented Weighted Estimator and Kernel Ridge Regression Application} \label{s:AWE}

\subsection{Construction of the Augmented Weighted Estimator}

In the preceding section, we constructed balancing weights $\hat{w}$ by minimizing a robust measure of covariate imbalance. A direct application of these weights yields a weighted estimator for $\boldsymbol{\beta}^*$. However, weighting-based estimators can exhibit substantial variance, particularly when the weights are highly variable due to regions of low treatment density \citep{robins2000marginal, kang2007demystifying}. 
We therefore introduce the Augmented Weighted Estimator (AWE), which incorporates an outcome regression term into the weighted estimating equation. This augmentation reduces variance while retaining the balancing properties of the weights. The construction is related to augmented minimax linear estimators \citep{hirshberg2020augmented} and more broadly to approaches that combine weighting with outcome modeling for variance reduction.

To align with the nonparametric structure of the framework, we estimate the outcome function using kernel ridge regression (KRR). Since the balancing procedure is defined over an RKHS, KRR provides a coherent choice of outcome model within the same function space. This shared structure also yields a computational advantage. The $n \times n$ Gram matrix $\mathbf{K}_1$ required for KRR is a submatrix of the $2n \times 2n$ kernel matrix $\tilde{\mathbf{K}}$ used to construct the balancing weights; see Appendix~\ref{app:formulas}, equation~\eqref{eq:k1_mat}. As a result, the most computationally intensive kernel evaluations are performed only once.

Given this setup, the KRR estimator of the outcome function $m(A,X)$ is obtained by solving
$$
    \min_{m \in \mathcal{H}} \left\{ \frac{1}{n} \sum_{i=1}^{n} \Bigl(Y_i - m(A_i, X_i) \Bigr)^2 + \eta \|m\|_{\mathcal{H}}^2 \right\},
$$
where $\eta > 0$ is a regularization parameter.
By the Representer Theorem, the solution evaluated at a new point $(A_{\text{new}}, X_{\text{new}})$ takes the form
$$    
    \hat{m}(A_{\text{new}}, X_{\text{new}}) = \sum_{j=1}^{n} c_j K\bigl( (A_j, X_j), (A_{\text{new}}, X_{\text{new}}) \bigr),
$$
where $(c_1, \dots, c_n)^\intercal = (\mathbf{K}_1 + \eta n \mathbf{I})^{-1} \mathbf{Y}$.
Using $\hat{m}$ we define the augmented pseudo-outcomes
\begin{equation}  
    \tilde{Y}_i = \hat{w}_i\left(Y_i - \hat{m}(A_i,X_i)\right) + \frac{1}{n} \sum_{j=1}^n \hat{m} (A_i, X_j),
    \label{eq:y_tilde_def}
\end{equation}
and let $\tilde{\mathbf{Y}} = (\tilde{Y}_1, \dots, \tilde{Y}_n)^\intercal$.
The first term applies the balancing weights to the residuals $Y_i - \hat{m}(A_i,X_i)$ rather than to the raw outcomes, isolating the unexplained variation. The second term adds a model-based estimate of the marginal mean over the covariates.


With these augmented pseudo-outcomes, estimation reduces to a least squares projection of $\tilde{\mathbf{Y}}$ onto the treatment compositions. The resulting estimator is
\begin{equation}
    \hat{\boldsymbol{\beta}}_{AWE} = (\mathbf{A}^\intercal \mathbf{A})^{-1} \mathbf{A}^\intercal \tilde{\mathbf{Y}}.
    \label{eq:awe_estimator}
\end{equation}

Although the AWE combines weighting with outcome regression, it differs from classical doubly robust estimators such as augmented inverse probability weighting (AIPW) and double machine learning. These methods typically rely on sample splitting or cross-fitting to control bias arising from flexible nuisance estimation, which can degrade finite-sample performance. In contrast, the proposed estimator uses the full sample for both weight construction and outcome estimation.
A further distinction lies in the conditions required for $\sqrt{n}$-consistency. Doubly robust estimators generally require consistent estimation of both the weight function and the outcome regression. In contrast, the proposed estimator imposes no smoothness requirement on the weight function, and its consistency does not depend on the convergence rate of the estimated weights.

\subsection{Selection of Penalty Parameter $\lambda$} \label{s:lambda}

The regularization parameter $\lambda$ controls the trade-off between reducing worst-case covariate imbalance and controlling estimator variance; see Section~\ref{s:Methodology}, \eqref{eq:w_hat_modified}. In practice, the optimal choice of $\lambda$ depends on unknown features of the data-generating process, so it must be selected using a data-driven criterion.

Our approach is based on the error decomposition in Section~\ref{s:Methodology}. The estimation error consists of a stochastic balancing error, an approximation error, and a stochastic noise term. The approximation error is deterministic and does not depend on the weights, so selecting $\lambda$ reduces to balancing the stochastic balancing and noise terms. However, both terms depend on the unknown outcome function $m$ and the unobserved noise $\boldsymbol{\varepsilon}$. 
We adopt a plug-in approach based on the KRR estimator $\hat{m}$ from the previous subsection. Specifically, we replace the unknown outcome function $m$ with $\hat{m}$ and the noise $\boldsymbol{\varepsilon}$ with the empirical residuals $\hat{\boldsymbol{\varepsilon}} = \mathbf{Y} - \hat{\mathbf{m}}$, where $\hat{\mathbf{m}} = (\hat{m}(A_1, X_1), \dots, \hat{m}(A_n, X_n))^\intercal$.

In practice, $\lambda$ is selected from a predefined candidate set $\Lambda$. For each $\lambda \in \Lambda$, we compute the weights $\hat{w}^{(\lambda)}$ by solving \eqref{eq:w_hat_modified} and define $\hat{\mathbf{W}}^{(\lambda)} = \mathrm{diag}(\hat w_1^{(\lambda)}, \dots, \hat w_n^{(\lambda)})$. The selected value $\hat{\lambda}$ minimizes the empirical error proxy
$$
	\hat{\lambda} = 
    \underset{\lambda \in \Lambda}{\text{argmin}} \,
    \left\{ \left\| (\mathbf{A}^\intercal \mathbf{A})^{-1} \mathbf{A}^\intercal \left( \hat{\mathbf{W}} ^{(\lambda)} \hat{\mathbf{m}} - \hat{\mathbb{E}}_{X} [ \hat{\mathbf{m}} ] \right) \right\|_2^2 + \left\| (\mathbf{A}^\intercal \mathbf{A})^{-1} \mathbf{A}^\intercal \hat{\mathbf{W}} ^{(\lambda)} \hat{\boldsymbol{\varepsilon}} \right\|_2^2 \right\}.
$$
The first term penalizes imbalance in the predicted outcome surface, and the second penalizes residual noise.

%
%
\section{Theoretical Properties} \label{s:theory_section}

This section presents the theoretical properties of the proposed balancing framework and the resulting AWE. We begin by stating the assumptions required for the analysis. Under these assumptions, we establish convergence properties of the proposed estimators. The results formalize how the balancing framework and tensor product RKHS structure together ensure stable estimation under complex confounding structures.

%
%
\subsection{Assumptions}

We first state the assumptions on the data-generating process.

\begin{assumption}[Positive Definite Second Moment] \label{ass:Sigma_A}
The second moment matrix $\boldsymbol{\Sigma}_A = \mathbb{E}[A A^\intercal]$ is strictly positive definite, meaning its minimum eigenvalue is strictly bounded away from zero; i.e., $\phi_{\min}(\boldsymbol{\Sigma}_A) \ge c_A > 0$. 
\end{assumption}

\begin{assumption}[Error Conditions] \label{ass:epsilon}
The error terms $\{\varepsilon_i\}_{i=1}^n$ are independent and satisfy $\mathbb{E}[\varepsilon_i \mid A_i, X_i] = 0$ and $\operatorname{Var}(\varepsilon_i \mid A_i, X_i) = \sigma_i^2$, where there exists a constant $\sigma^2 < \infty$ such that $\sigma_i^2 \le \sigma^2$ for all $i=1, \dots, n$. 
\end{assumption}

\begin{assumption}[Bounded Weights] \label{ass:C_w}
The true weight function $w^*$ is uniformly bounded. That is, there exists a constant $C_w < \infty$ such that $\sup_{(a, x) \in \mathcal{A} \times \mathcal{X}} w^*(a, x) \leq C_w$. 
\end{assumption}

These assumptions are commonly employed in causal inference settings \citep{newey1994asymptotic, hirano2003efficient, graham2012inverse, ichimura2022influence}.
Assumption~\ref{ass:Sigma_A} guarantees identifiability through sufficient treatment variation. 
Assumption~\ref{ass:epsilon} ensures a well-defined conditional mean and controlled variability. 
Assumption~\ref{ass:C_w} is closely related to the overlap condition introduced in Section~\ref{s:preliminaries}, preventing extreme reweighting. 

We now state the assumptions on the outcome model.

\begin{assumption}[RKHS Structure with Bounded Kernels] \label{ass:kernels}
The true mean response function $m$ belongs to a Reproducing Kernel Hilbert Space $\mathcal{H}$ defined by a tensor product kernel 
$$ K((A,X),(A',X')) = K_A(A,A')K_X(X,X'). $$
The kernels $K_A$ and $K_X$ are uniformly bounded with finite constants $\kappa_A, \kappa_X < \infty$, i.e., 
$$\sup_{A \in \mathcal{A}} K_A(A, A) \le \kappa_A^2 \quad \text{and} \quad \sup_{X \in \mathcal{X}} K_X(X, X) \le \kappa_X^2.$$
\end{assumption}

\begin{assumption}[Entropy Condition] \label{ass:entropy}
The uniform entropy of the unit ball $\mathcal{H}(1)$ is bounded, such that for some constants $B>0$ and $0 < \alpha < 2$, we have $H(\mathcal{H}(1), \|\cdot\|_{\infty}, \epsilon) \leq B\,\epsilon^{-\alpha}$ for all $\epsilon > 0$. 
\end{assumption}

Assumption~\ref{ass:kernels} is mild and satisfied by many commonly used kernels such as the Gaussian kernel, which is bounded \citep{scholkopf2002learning, steinwart2008support}.
Assumption~\ref{ass:entropy} is a weak condition that holds for many commonly used function classes, including Sobolev spaces and RKHS induced by Gaussian kernels \citep{gyorfi2002distribution, zhou2002covering}.

%
%
\subsection{Theoretical Results}

We now present the main theoretical results, which establish statistical consistency and convergence rates for both the kernel balancing estimator and the AWE. 
We begin with a benchmark result for the estimator $\hat{\boldsymbol{\beta}}$ constructed using kernel functional balancing weights alone.

\begin{theorem}[Convergence Rate of the Weighted Estimator] \label{thm:conv_rate}
    Suppose Assumptions~\ref{ass:Sigma_A}--\ref{ass:entropy} hold. Let
    $$
        \hat{\boldsymbol{\beta}}
        =
        (\mathbf{A}^\intercal \mathbf{A})^{-1} \mathbf{A}^\intercal \hat{\mathbf{W}} \mathbf{Y},
    $$
    where $\hat{\mathbf{W}} = \mathrm{diag}(\hat{w})$ and $\hat{w}$ is defined in \eqref{eq:w_hat_def}.
    Then
    $$
        \left\| \hat{\boldsymbol{\beta}} - \boldsymbol{\beta}^* \right\|_2
        =
        O_p \left( \frac{1}{\sqrt{n}} \big( \| m \|_{\mathcal{H}} C_1 + \| m \|_\infty C_2 + C_3 \big) \right),
    $$
    
    where $C_1$, $C_2$, and $C_3$ are constants determined by those introduced in Assumptions~\ref{ass:Sigma_A}--\ref{ass:entropy} and the regularization parameter $\lambda \asymp 1$. Their explicit expressions are provided in \eqref{eq:thm3_C_1}, \eqref{eq:thm3_C_2}, and \eqref{eq:thm3_C_3}. 
\end{theorem}

This result establishes consistency of the proposed weighted estimator.
Its convergence rate depends on $\|m\|_{\mathcal H}$, which reflects the complexity of the outcome function.
When the response surface is highly complex, this dependence may deteriorate finite-sample performance.
The following theorem shows that the augmented estimator replaces this dependence with the estimation error of $\hat m$.

\begin{assumption}[Convergence Rate of the Outcome Function Estimation Error] \label{ass:delta_conv}
    Let $m$ denote the true outcome function and $\hat{m}$ its estimator, both belonging to the Reproducing Kernel Hilbert Space $\mathcal{H}$. There exists a constant $\zeta > 0$ such that 
    $$ \| \hat{m} - m \|_{\mathcal{H}} = O_p\bigl(n^{-\zeta}\bigr).$$
\end{assumption}

\begin{theorem}[Convergence Rate of the Augmented Weighted Estimator] \label{thm:awe_rate}
    Under the conditions of Theorem~\ref{thm:conv_rate} and  Assumption~\ref{ass:delta_conv}, let $\hat{\boldsymbol{\beta}}_{AWE}$ be defined as in \eqref{eq:awe_estimator}. 
    Then
    $$
        \left\| \hat{\boldsymbol{\beta}}_{AWE} - \boldsymbol{\beta}^* \right\|_2
        =
        O_p \left( \frac{1}{\sqrt{n}} \big( n^{-\zeta} C_1 + \| m \|_\infty C_2 + C_3 \big) \right),
    $$
    where $C_1$, $C_2$, and $C_3$ are constants determined by those introduced in Assumptions~\ref{ass:Sigma_A}--\ref{ass:entropy} and the regularization parameter $\lambda \asymp 1$. Their explicit expressions are provided in \eqref{eq:thm3_C_1}, \eqref{eq:thm3_C_2}, and \eqref{eq:thm3_C_3}.
\end{theorem}

Theorem~\ref{thm:awe_rate} replaces the dependence on $\|m\|_{\mathcal H}$ with a term involving $\|\hat m - m\|_{\mathcal H}$.
As the estimation error decreases, the contribution of the outcome regression term becomes smaller even when the true outcome function is complex.

\begin{remark}
    Assumption~\ref{ass:delta_conv} is satisfied when $\hat{m}$ is estimated via KRR. In particular, Theorem~1(ii) of \citet{fischer2020sobolev} implies that, under standard regularity conditions on the kernel and the target function, including eigenvalue decay and source conditions, and with an appropriate choice of the regularization parameter of KRR $\eta$, $\hat{m}$ satisfies the assumption with some $\zeta > 0$.
\end{remark}

We next study the asymptotic distribution of the proposed estimator through the decomposition
$$
	\hat{\boldsymbol{\beta}}_{AWE} - \boldsymbol{\beta}^* = (\hat{\boldsymbol{\beta}}_{AWE} - \tilde{\boldsymbol{\beta}}) + (\tilde{\boldsymbol{\beta}} - \boldsymbol{\beta}^*).
$$
where $ \tilde{\boldsymbol{\beta}} = (\mathbf{A}^\intercal \mathbf{A})^{-1} \mathbf{A}^\intercal \hat{\mathbb{E}}_X[\mathbf{m}].$ 
The quantity $\tilde{\boldsymbol{\beta}}$ represents the sample-specific linear projection induced by the empirical covariate distribution under the true outcome function $m$.
In many causal inference procedures, asymptotic normality with respect to the target parameter $\boldsymbol{\beta}^*$ typically requires consistent estimation of both the outcome regression function and the treatment assignment mechanism \citep{robins1994estimation, bang2005doubly}. 
These results often rely on convergence rate conditions and additional regularity assumptions on the nuisance estimators.
In contrast, the proposed estimator does not require the estimated weights $\hat{w}$ to converge to the true weights $w^*$, nor does it impose convergence rate or smoothness conditions on the weight estimator.
While this avoids imposing strong assumptions on the weighting mechanism, it also makes it difficult to characterize the asymptotic variance of $ \hat{\boldsymbol{\beta}}_{AWE} - \boldsymbol{\beta}^* $ through convergence to a fixed covariance matrix.
By comparison, the variance of $ \hat{\boldsymbol{\beta}}_{AWE} - \tilde{\boldsymbol{\beta}} $ can be characterized directly. 
This allows us to establish asymptotic normality for the intermediate term without imposing additional strong assumptions on the estimated weights.

\begin{assumption}[Conditional Third Moment Condition] \label{ass:eps_third_moment}
    Let $\varepsilon_i = Y_i - m(A_i, X_i)$ denote the true residual.
    There exists a constant $C_\varepsilon > 0$ such that
    $$ 
        \sup_{1 \le i \le n} \mathbb{E} \left[ | \varepsilon_i |^3 \mid A_i, X_i \right] \le C_\varepsilon
    $$
    almost surely.
\end{assumption}

\begin{assumption}[Weighted Stability Condition]
\label{ass:weighted_stability}
The estimated weights $\hat{w}$ satisfy
$$
    \| \hat{w} \|_{\infty} = o_p(n^{1/6}),
$$
and there exists a constant $c_V > 0$ such that
$$
    \phi_{\min} \left( 
    \frac{1}{n} \sum_{i=1}^n \hat{w}_i^2 \sigma_i^2 A_i A_i^\intercal
    \right) \ge c_V
$$
with probability tending to one.
\end{assumption}

Assumption~\ref{ass:weighted_stability} is imposed to control the weighted noise term in the asymptotic normality proof. The growth condition on $\|\hat{w}\|_\infty$ is used to verify the Lyapunov condition, while the nondegeneracy condition ensures that the asymptotic covariance matrix remains well-conditioned.
A stronger sufficient condition for the nondegeneracy requirement is the existence of constants $c_w > 0$ and $c_\sigma > 0$ such that
$$
\hat{w}_i \ge c_w
\quad \text{and} \quad
\sigma_i^2 \ge c_\sigma
$$
for all $i = 1, \dots, n$ with probability tending to one. Combined with Assumption~\ref{ass:Sigma_A}, these conditions imply the weighted covariance condition in Assumption~\ref{ass:weighted_stability} asymptotically.

\begin{theorem}[Asymptotic Normality of the Augmented Weighted Estimator] \label{thm:awe_normality}
    Under the conditions of Theorem~\ref{thm:awe_rate} and Assumptions~\ref{ass:eps_third_moment} and \ref{ass:weighted_stability}, it holds that
    $$
        \sqrt{n} 
        \left\{ 
        \hat{\boldsymbol{\Sigma}}_A^{-1} \left( \frac{1}{n} \sum_{i=1}^n \hat{w}_i^2 \hat{\varepsilon}_i^2 A_i A_i^\intercal \right) \hat{\boldsymbol{\Sigma}}_A^{-1} 
        \right\}^{-1/2} 
        \left( 
        \hat{\boldsymbol{\beta}}_{AWE} - \tilde{\boldsymbol{\beta}} 
        \right) 
        \overset{d}{\longrightarrow} \mathcal{N}(\mathbf{0}, \mathbf{I}),
    $$
    where $ \hat{\varepsilon}_i = Y_i - \hat{m} (A_i, X_i) $.
\end{theorem}

Theorem~\ref{thm:awe_normality} shows that the leading stochastic component of the augmented weighted estimator admits an asymptotically Gaussian representation after appropriate normalization. The normalization matrix is constructed using the estimated residuals $\hat{\varepsilon}_i^2$ in place of the unknown conditional variances $\sigma_i^2$, which yields a fully data-driven standardization. The result establishes asymptotic normality around the sample-dependent quantity $\tilde{\boldsymbol{\beta}}$, rather than the population target $\boldsymbol{\beta}^*$. Under the proposed normalization, the limiting distribution is centered at zero with identity covariance matrix.

%
%
\section{Simulation Studies} \label{s:simulation}

To evaluate the finite-sample performance of the proposed AWE built upon the kernel functional balancing framework, we conduct a comprehensive simulation study. The objective is to examine robustness and efficiency under data generating processes that exhibit different structural relationships among the compositional treatment, covariates, and outcome. We assess the variance-regularizing weighting scheme and the augmentation step by comparing the proposed approach with several alternative estimators across three distinct scenarios.

\subsection{Data Generating Process}

Synthetic datasets of size $n$ are generated with treatment and covariate dimensions fixed at $L=p=3$. The covariates $X_i = (x_{i1}, x_{i2}, x_{i3})^\intercal$ are drawn from a multivariate normal distribution
$$
    X_i \sim \mathcal{N}(\mathbf{1}, \boldsymbol{\Sigma}_X),
    \text{ where }
    \boldsymbol{\Sigma}_X = \begin{pmatrix}
        1 & 0.4 & 0.4 \\
        0.4 & 1 & 0.4 \\
        0.4 & 0.4 & 1
    \end{pmatrix},
$$
which induces moderate correlation among covariates.

To introduce nonlinear confounding between $X_i$ and treatment allocation, unnormalized treatment components $A_i' = (a_{i1}', a_{i2}', a_{i3}')^\intercal$ are generated through nonlinear transformations of the covariates
\begin{align*}
    a_{i1}' \sim \mathcal{N}(4 \sqrt{|x_{i1}|}, 1), \quad
    a_{i2}' \sim \mathcal{N}(1.5 x_{i2}, 1), \quad
    a_{i3}' \sim \mathcal{N}(0.5 x_{i3}^2, 1).
\end{align*}
The treatment vector $A_i$ is obtained by normalizing the absolute values of these components onto the $(L-1)$-dimensional simplex
$$ 
    A_i = \left( \frac{|a_{i1}'|}{\sum_{l=1}^L |a_{il}'|}, \dots, \frac{|a_{iL}'|}{\sum_{l=1}^L |a_{il}'|} \right)^\intercal. 
$$

The outcome is generated according to
$$ 
    Y_i = m(A_i, X_i) + \varepsilon_i, \qquad \varepsilon_i \sim \mathcal{N}(0, 1)
$$
under three structural models
\begin{align*}
    \text{Model 1 (Linear, Additive): } & m_1(A_i, X_i) = A_i^\intercal \boldsymbol{\beta} + X_i^\intercal \boldsymbol{\gamma}, \\
    \text{Model 2 (Nonlinear, Additive): } & m_2(A_i, X_i) = A_i^\intercal \boldsymbol{\beta} + g(X_i) + 2, \\
    \text{Model 3 (Nonlinear, Interactive): } & m_3(A_i, X_i) = (A_i^\intercal \boldsymbol{\beta}) \{ g(X_i) + 2 \},
\end{align*}
where $g(X_i)$ is the centered transformation of $x_{i1}^3 + x_{i2}x_{i3}$.

The vector $\boldsymbol{\beta}$ determines the contribution of the compositional treatment components in the outcome model, while $\boldsymbol{\gamma}$ governs the covariate contribution in Model 1. In the simulation, each component of $\boldsymbol{\beta}$ is drawn independently as an integer from $\{-5,-4,\dots,-1,1,\dots,5\}$, excluding zero. Each component of $\boldsymbol{\gamma}$ is drawn independently from $\{1,2,\dots,5\}$.

Because the parameter of interest is the population-level causal estimand $\boldsymbol{\beta}^*$ rather than the structural outcome coefficient $\boldsymbol{\beta}$, we approximate $\boldsymbol{\beta}^*$ using Monte Carlo integration over the marginal distribution of $X$. By generating $B$ independent draws $X^{(b)}$ from the marginal distribution of $X$, the target parameter under Model 1 is numerically evaluated as
$$
    \boldsymbol{\beta}^* \approx \boldsymbol{\beta} + \left( \frac{1}{B}\sum_{b=1}^{B} \left(X^{(b)}\right)^\intercal \boldsymbol{\gamma} \right) \mathbf{1}_L,
$$
where $\mathbf{1}_L$ is an $L$-dimensional vector of ones. 
Similarly, under Models 2 and 3, the relationships are approximated by
$$
    \boldsymbol{\beta}^* \approx \boldsymbol{\beta} + \left( \frac{1}{B}\sum_{b=1}^{B} \left\{ g\left(X^{(b)}\right) + 2 \right\} \right) \mathbf{1}_L, \qquad
    \boldsymbol{\beta}^* \approx \boldsymbol{\beta} \left( \frac{1}{B}\sum_{b=1}^{B} \left\{ g\left(X^{(b)}\right) + 2 \right\} \right),
$$
respectively. The number of Monte Carlo replicates is set to $B = 10n$.

\subsection{Estimators and Comparison Framework}

We evaluate eight estimators defined by different combinations of outcome modeling and weighting strategies. The outcome model $\hat{m}$ is used for selecting the tuning parameter $\lambda$ and constructing the augmented outcome vector $\tilde{\mathbf{Y}}$ required for the AWE procedure. As indicated in \eqref{eq:w_hat_modified}, $\hat{m}$ is not used in estimating the weights $\hat{w}$.

Four outcome models are considered. The naive model,
$ \hat{m}(A) = A^\intercal \boldsymbol{\beta}, $
serves as a simple parametric baseline that omits covariates $X$.  
The linear model with covariates,
$ \hat{m}(A, X) = A^\intercal \boldsymbol{\beta} + X^\intercal \boldsymbol{\gamma}, $
provides a standard linear regression model.
KRR, $ \hat{m}(A, X) = \hat{f}_{\mathrm{KRR}}(A, X)$, is a flexible nonparametric estimator described in Section~\ref{s:AWE}.  
Finally, the smoothing spline model,
$ \hat{m}(A, X) = \hat{f}_{\mathrm{SS}}(A, X), $
assumes an additive structure without interaction effects.

Three weighting strategies are examined. Uniform weights set $\hat{w}_i = 1$ and serve as an unweighted baseline without covariate balancing. The proposed penalized weights are obtained from the objective in \eqref{eq:w_hat_modified} and are designed to achieve covariate balance. Dirichlet regression weights are derived from a parametric Dirichlet model and represent a conventional approach for compositional treatments.

Combining these components yields eight estimators. Their definitions and labeling conventions are summarized in Table~\ref{tab:estimator_mapping}.
\begin{table}[htbp]
\centering
\small
\caption{Estimators used in Simulation Studies.}
\label{tab:estimator_mapping}
\begin{tabular}{l l l}
    \toprule
    \textbf{Estimator Label} & \textbf{Outcome Model ($\hat{m}$)} & \textbf{Weighting Method ($\hat{w}$)} \\
    \midrule
    NM         & Naive Model        & Uniform (Unweighted) \\
    Weighted   & -                  & Proposed Penalized Weights \\
    Dirichlet  & -                  & Dirichlet Regression Weights \\
    \addlinespace
    LM         & Linear Model with Covariates  & Uniform (Unweighted) \\
    KRR        & Kernel Ridge Regression       & Uniform (Unweighted) \\
    SS         & Smoothing Spline              & Uniform (Unweighted) \\
    \addlinespace
    KRR (AWE)  & Kernel Ridge Regression       & Proposed Penalized Weights (with Augmentation) \\
    KRR (Diri) & Kernel Ridge Regression       & Dirichlet Regression Weights (with Augmentation) \\
    \bottomrule
\end{tabular}
\end{table}

For each estimator, the target parameter $\boldsymbol{\beta}^*$ is obtained by 
$$ 
    \hat{\boldsymbol{\beta}}^* = \left(\mathbf{A}^\intercal \mathbf{A}\right)^{-1} \mathbf{A}^\intercal \mathbf{Y}^{\dagger},
$$
where $\mathbf{Y}^{\dagger} = (Y_1^{\dagger}, \ldots, Y_n^{\dagger})^\intercal$.
The form of $\mathbf{Y}^{\dagger}$ depends on whether augmentation is employed. For unaugmented estimators, $\mathbf{Y}^{\dagger} = \hat{\mathbf{W}} \mathbf{Y}$, whereas for augmented estimators, $\mathbf{Y}^{\dagger} = \tilde{\mathbf{Y}}$. The weights $\hat{w}_i$ are obtained using one of the three methods described above.

\subsection{Results and Discussion}

\begin{table}[htbp]
\centering
\footnotesize
\setlength{\tabcolsep}{4pt}
\caption{Simulation results under $m_1$, $m_2$, and $m_3$}
\label{tab:mse_results}
\begin{tabular}{l c c c | c c c | c c c}
    \toprule
    & \multicolumn{3}{c|}{Results under $m_1$} 
    & \multicolumn{3}{c|}{Results under $m_2$} 
    & \multicolumn{3}{c}{Results under $m_3$} \\
    \cmidrule(lr){2-10}
    \textbf{Model} & $n=100$ & $n=200$ & $n=400$ & $n=100$ & $n=200$ & $n=400$ & $n=100$ & $n=200$ & $n=400$ \\
    \midrule
    NM         & 27.8649 & 18.0018 & 14.7202 & 14.5455 & 7.3647 & 5.1677 & 43.1794 & 37.7537 & 33.7000 \\
    Weighted   & 20.3408 & 13.4687 & 12.5837 & 6.2250 & 4.0681 & 2.2839 & 28.5825 & 26.0877 & 23.6755 \\
    Dirichlet  & 2154.1743 & 1533.0190 & 1088.9011 & 351.6335 & 373.3823 & 194.5801 & 574.7100 & 575.2398 & 530.5572 \\
    \hline
    LM         & 0.7419 & 0.4060 & 0.1953 & 8.8508 & 4.6514 & 3.9442 & 119.0623 & 110.4890 & 110.0103 \\
    KRR        & 6.4389 & 4.3733 & 3.2061 & 4.2973 & 3.4678 & 2.7087 & 34.4532 & 27.0316 & 22.2884 \\
    SS         & 0.7896 & 0.4194 & 0.1994 & 0.7549 & 0.3064 & 0.1755 & 109.8386 & 101.3640 & 93.5561 \\
    \hline
    KRR (AWE)  & 5.5241 & 3.7023 & 2.7915 & 4.2066 & 3.3803 & 2.6367 & 33.7203 & 26.5738 & 21.9904 \\
    KRR (Diri) & 16.7829 & 11.6118 & 7.3741 & 8.7160 & 6.1644 & 3.8063 & 42.0677 & 28.0540 & 21.8820 \\
    \bottomrule
\end{tabular}\\
All entries represent the Mean Squared Error (MSE).
\end{table}

Model 1 considers a linear and additive outcome structure, serving as a benchmark setting where linear regression is correctly specified. As expected, the unweighted estimators `LM' and `SS' perform well, with `LM' achieving the best performance since it coincides with the true outcome model. Unsurprisingly, `NM' yields unsatisfactory results. While `Weighted' offers a slight improvement over `NM', the `KRR' estimator exhibits stable performance, with its MSE steadily decreasing as the sample size increases. Building on this, `KRR (AWE)' further improves upon the performance of unweighted `KRR'. In contrast, `KRR (Diri)' degrades the initial `KRR' estimates. Notably, traditional weighting via Dirichlet regression without the augmentation step (`Dirichlet') leads to highly unstable and explosive MSE values.

Model 2 introduces a nonlinear additive structure through $g(X)$. While `LM' is misspecified with respect to the components related to $X$, the `SS' estimator provides a close approximation to the true model and naturally yields the best results. In this scenario, `NM' and `Weighted' generally achieve better results than in Model 1. Notably, as the underlying model becomes more complex, `Weighted' exhibits a much more pronounced performance improvement than `NM'. The `KRR' estimator remains robust and performs well, even showing slight improvement relative to its performance in Model 1. In this setting, `KRR (AWE)' successfully refines the `KRR' estimates, yielding a modest improvement. However, Dirichlet regression related methods once again exhibit poor performance.

Model 3 introduces interactions between $A$ and $X$, violating the additivity assumptions of `LM' and `SS', which leads to persistent large errors for these models. In this complex setting, overall MSEs increase and the performance of `NM' deteriorates significantly. In contrast, `Weighted' maintains strong performance comparable to `KRR' and augmented models. `KRR' consistently outperforms the misspecified parametric models, and `KRR (AWE)' further improves these estimates. While `KRR (Diri)' anomalously achieves the lowest error at the largest sample size, this is merely an incidental artifact. Dirichlet regression requires strictly positive components, and the necessary ad-hoc adjustments for zero values in $A$ distort original treatment ratios, making the estimation unreliable despite the seemingly favorable MSE.

Overall, the simulation study demonstrates the versatility and robustness of the proposed `KRR (AWE)' estimator. Although `KRR (AWE)' produces higher MSEs than `LM' or `SS' when those specific parametric forms perfectly match the true data generating process, it consistently achieves the best performance among the remaining estimators across all scenarios. In contrast, traditional parametric weighting strategies frequently produce unstable weights and can substantially degrade estimation accuracy.

%
%
\section{Application to Real Data} \label{s:real_data}

We applied the proposed methodology to the 2024 \emph{American Time Use Survey} (ATUS), a nationally representative dataset that records how individuals allocate their time across various daily activities. Our final sample consists of $n = 1,\!868$ observations. In this dataset, the treatment variable $A$ corresponds to the proportion of time devoted to different categories of activities, while the outcome variable $Y$ is defined as each individual’s raw weekly income.

We include the demographic and socio-economic covariates \emph{household size}, \emph{number of children}, \emph{age}, and \emph{education level}, which we collectively denote by $X$.
\begin{figure}[htbp] 
    \centering
    
    \begin{subfigure}{0.48\textwidth}
        \centering
        \includegraphics[width=0.9\linewidth]{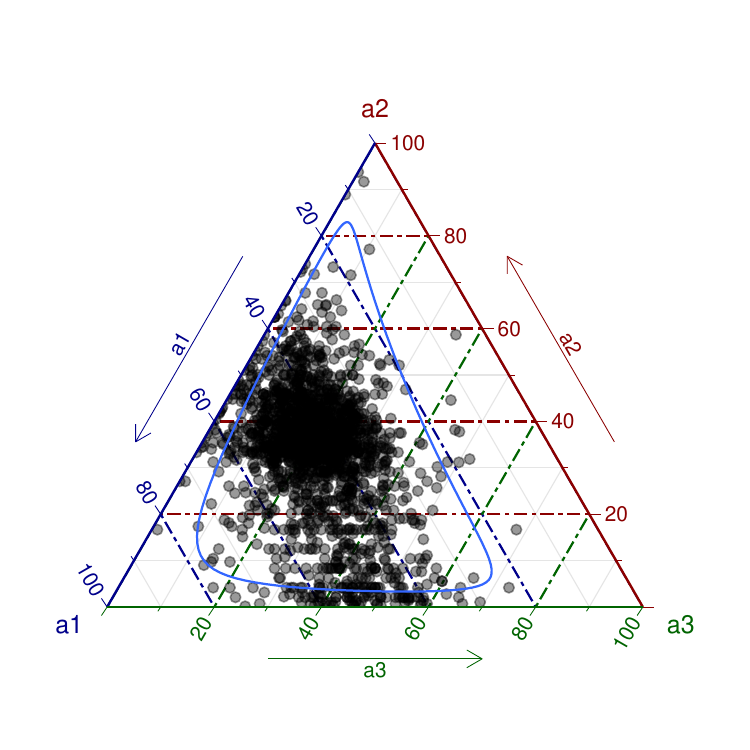}
        \caption{Observed compositions}
    \end{subfigure}
    \hfill
    \begin{subfigure}{0.48\textwidth}
        \centering
        \includegraphics[width=1.18\linewidth]{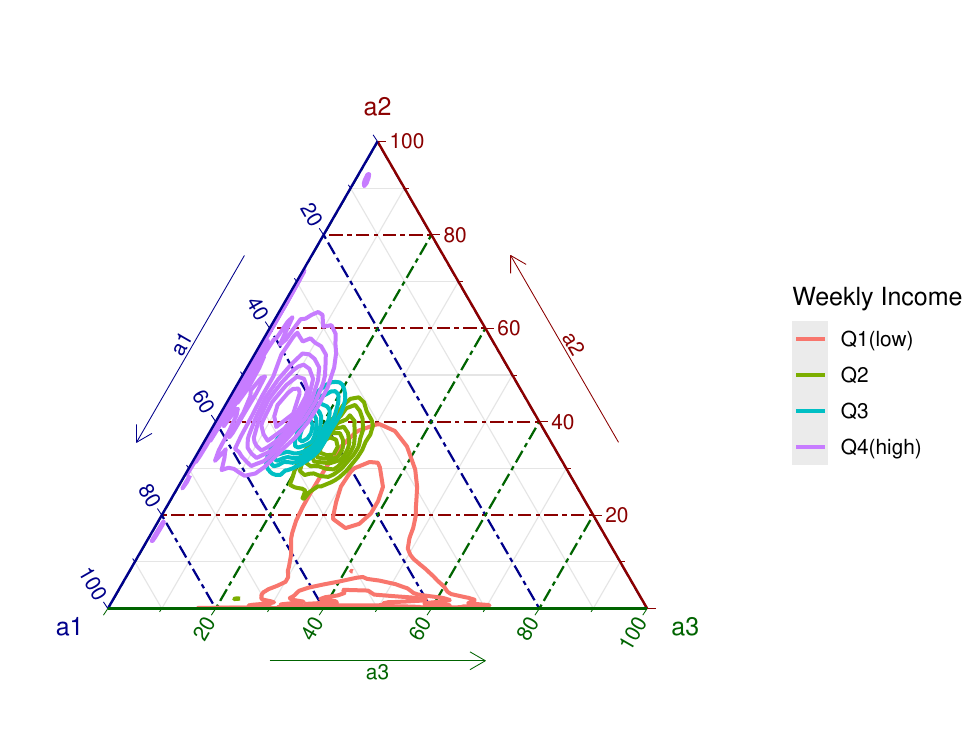}
        \caption{Predicted group distributions}
    \end{subfigure}

    \vspace{0.6em}

    \begin{subfigure}{0.48\textwidth}
        \centering
        \includegraphics[width=\linewidth]{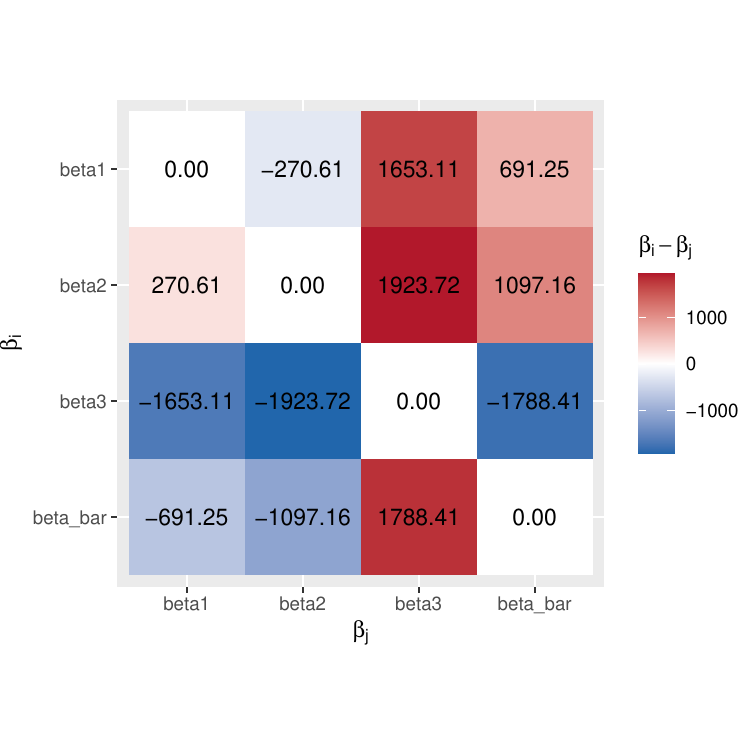}
        \caption{Relative shift effects (3D)}
    \end{subfigure}
    \hfill
    \begin{subfigure}{0.48\textwidth}
        \centering
        \includegraphics[width=\linewidth]{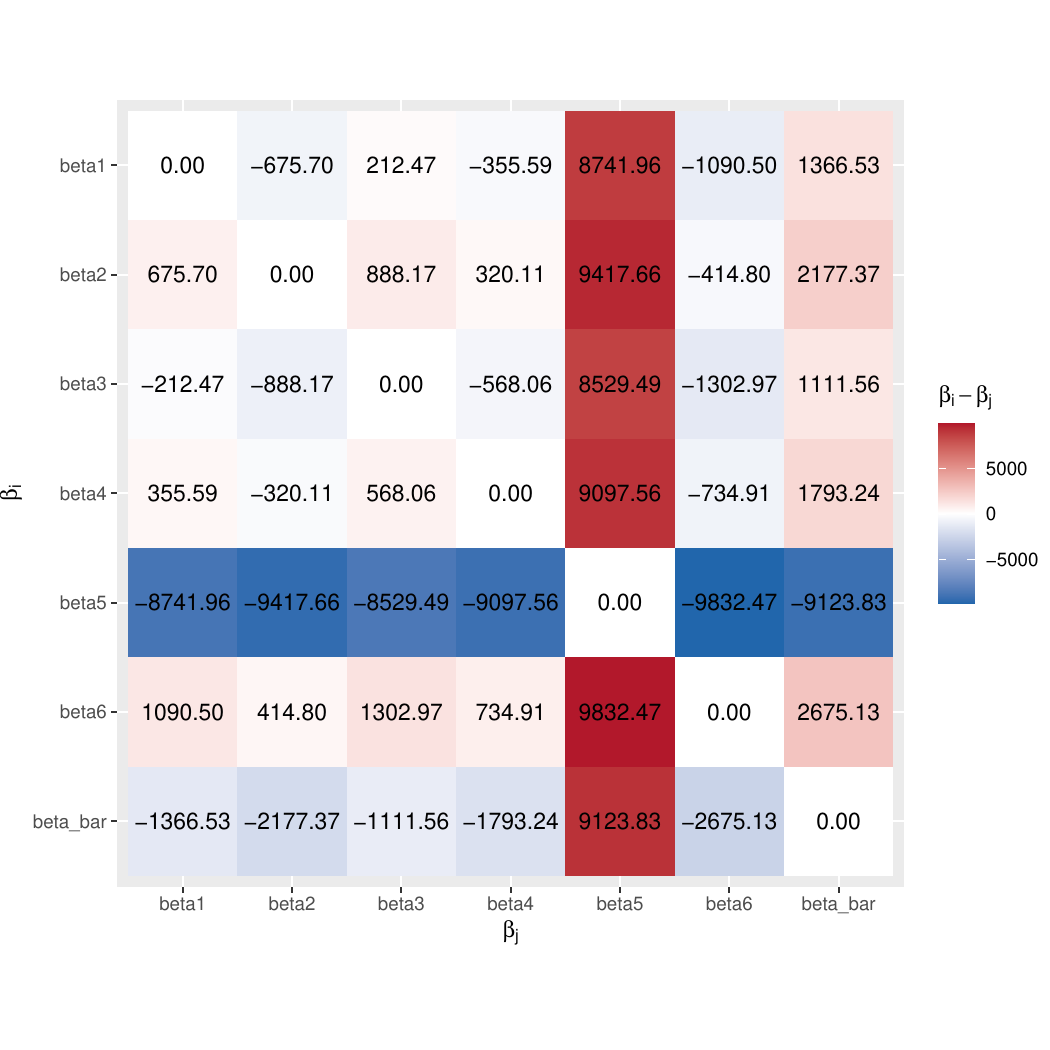}
        \caption{Relative shift effects (6D)}
    \end{subfigure}

    \caption{
        Time-use compositions and KRR(AWE)-based results.
        (a) Observed compositions. \\
        (b) Predicted group distributions.
        (c) Relative shift effects (3D).
        (d) Relative shift effects (6D).
    }
    \label{fig:results_KRR(AWE)}
\end{figure}

Figure~\ref{fig:results_KRR(AWE)}(a) displays the empirical distribution of the treatment composition $A = (a_1, a_2, a_3)^\intercal$ on the simplex. Each point represents an individual's daily time allocation, with components summing to one. For interpretability, the original activity categories are aggregated into three components: maintenance ($a_1$, including sleeping, eating, and personal care), productive ($a_2$, including working and studying), and discretionary ($a_3$, including leisure, socializing, and entertainment), as defined in Appendix~\ref{app:activity_definition}.  
Most observations lie in the interior of the simplex, with noticeable concentration around moderate shares of maintenance time. Across the sample, $a_1$ accounts for approximately 40--60\% of total time, while $a_2$ and $a_3$ exhibit greater variability.

Figure~\ref{fig:results_KRR(AWE)}(b) presents the predicted group-wise activity distributions. 
For each individual, we compute
$$ \hat{Y}_i = A_i^\intercal \hat{\boldsymbol{\beta}}, $$
where $\hat{\boldsymbol{\beta}}$ is estimated using the proposed method. 
The sample is partitioned into four groups according to the empirical quartiles of $\{\hat{Y}_i\}_{i=1}^n$.
The simplex shows group-wise compositions across increasing predicted outcome levels. 
A monotonic pattern is observed. 
The proportion of time spent on $a_1$ remains stable across all income groups, and the primary factor distinguishing the quartiles is $a_2$ and $a_3$. 
Higher-income groups allocate more time to $a_2$ and less time to $a_3$, whereas the lowest-income group allocates the opposite. 
Q2 and Q3 have nearly identical compositions, suggesting limited separation among middle-income groups. 
Furthermore, the lowest-income group allocates a remarkably larger share to $a_3$ than the other groups.

Figure~\ref{fig:results_KRR(AWE)}(c) shows the estimated reallocation effects $\tau(a',a)$ defined in \eqref{eq:tau}. 
Each cell represents the expected change in weekly income when one hour of time is reallocated from the $j$-th activity (horizontal axis) to the $i$-th activity (vertical axis), holding total time fixed. 
Darker red cells indicate larger positive effects, and darker blue cells indicate larger negative effects. 
The fourth column reports $\tau(\bar{a}', a)$, corresponding to a joint reallocation in which one hour is shifted to activity $i$ from the remaining activities.
The estimates exhibit consistent patterns. 
Reallocating time toward $a_2$ is associated with higher income, particularly when time is shifted from maintenance or discretionary activities. 
Reducing $a_3$ is also associated with higher income, even when the time is reallocated to $a_1$. 
The fourth column summarizes these joint reallocations. 
In particular, reallocations toward $a_2$ yield the largest increases in weekly income, 
whereas reallocations toward $a_3$ yield the largest decreases.

Figures~\ref{fig:app_prediction_plot} and \ref{fig:app_reallocation_panels} in Appendix~\ref{app:figures} present the corresponding results obtained using alternative methods for Figures~\ref{fig:results_KRR(AWE)} (b) and (c), respectively. 
The results from (d) LM and (e) KRR follow similar patterns, although the separation is less distinct. 
(a) NM and (f) KRR(Diri) produce patterns with reversed directions, in contrast to the other methods.
Additional details are provided in Appendix~\ref{app:figures}.

Figure~\ref{fig:results_KRR(AWE)}(d) reports the same analysis under a six-dimensional treatment composition, with activity categories expanded from three to six; definitions are given in Appendix~\ref{app:activity_definition}. 
Reallocations toward $a_1$ and $a_2$ are associated with increases in weekly income, whereas shifts toward $a_3$ and $a_5$ lead to decreases. 
The effects are more dispersed than in the three-dimensional case, but the directional patterns remain consistent. 
Overall, the results show that the proposed estimator produces stable and interpretable reallocation effects across specifications.


\bibliographystyle{apalike}
\bibliography{references}

\appendix
\setcounter{theorem}{0}
\section{Additional Materials}
%
%
\subsection{Additional Formulations}
\label{app:formulas}

This section provides the explicit block-matrix form of the $2n \times 2n$ kernel matrix $\tilde{\mathbf{K}}$ introduced in Theorem~\ref{thm:representer}. The matrix is partitioned as
\begin{align}
    \tilde{\mathbf{K}} = \begin{pmatrix} \mathbf{K}_1 & \mathbf{K}_2 \\ \mathbf{K}_2^\intercal & \mathbf{K}_3 \end{pmatrix},
    \label{eq:k_tilde_mat}
\end{align}
where the submatrices $\mathbf{K}_1$, $\mathbf{K}_2$, and $\mathbf{K}_3$ are defined as follows. The matrix $\mathbf{K}_1$ is the $n \times n$ kernel matrix, which is also utilized for the outcome function estimation in AWE, 
\begin{align}
    \mathbf{K}_1 = 
        \begin{pmatrix} 
            K_{A11}K_{X11} & \cdots & K_{A1n}K_{X1n} \\ 
            \vdots & \ddots & \vdots \\ 
            K_{An1}K_{Xn1} & \cdots & K_{Ann}K_{Xnn}
        \end{pmatrix}.
    \label{eq:k1_mat}
\end{align}
The submatrices $\mathbf{K}_2$ and $\mathbf{K}_3$ are given by
\begin{align*}
    \mathbf{K}_2 = 
        \begin{pmatrix} 
            \mu(K_{X1})K_{A11} & \cdots & \mu(K_{X1})K_{A1n} \\ 
            \vdots & \ddots & \vdots \\ 
            \mu(K_{Xn})K_{An1} & \cdots & \mu(K_{Xn})K_{Ann}
        \end{pmatrix}, \quad
    \mathbf{K}_3 = \mu(K_X) 
        \begin{pmatrix} 
            K_{A11} & \cdots & K_{A1n} \\ 
            \vdots & \ddots & \vdots \\ 
            K_{An1} & \cdots & K_{Ann}
        \end{pmatrix}.
\end{align*}
The elements in these matrices are defined using the following notation:
\begin{itemize}
    \item $K_{Aij} = K_A(A_i, A_j)$ and $K_{Xkl} = K_X(X_k, X_l)$,
    \item $\mu(K_{Xl}) = \frac{1}{n} \sum_{k=1}^n K_{Xkl}$,
    \item $\mu(K_X) = \frac{1}{n^2} \sum_{k=1}^n \sum_{l=1}^n K_{Xkl}$.
\end{itemize}
%
%
\subsection{Definition of Aggregated Activity Components}
\label{app:activity_definition}
For interpretability, the original activity categories are aggregated into three broad components defined as follows.
\begin{itemize}
    \item Maintenance Activities ($a_1$): Essential daily survival and upkeep tasks, including sleep, eating and drinking, and purchasing goods and services.

    \item Productive Activities ($a_2$): Activities directly linked to human capital accumulation and income generation, such as work and education.

    \item Discretionary Activities ($a_3$): Leisure, social engagement, and other miscellaneous or residual activities. This includes socializing, relaxation, sports, volunteer work, and religious observance.
\end{itemize}

\noindent
The activities can alternatively be divided into six categories defined as follows.
\begin{itemize}
    \item Maintenance Activities ($a_1$): Essential daily survival and upkeep tasks, including sleep, eating and drinking.
    \item Productive Activities ($a_2$): Activities directly linked to human capital accumulation and income generation, such as work and education
    \item Consumption Activities ($a_3$): Activities related to purchasing goods and services.
    \item Leisure Activities ($a_4$): Socializing, leisure, sports, exercise, and recreational activities.
    \item Civic and Religious Activities ($a_5$): Nonmarket social engagement, including religious, spiritual, and volunteer activities.
    \item Residual Discretionary Activities ($a_6$): Other discretionary time uses not captured above, such as telephone calls and related minor activities.
\end{itemize}
%
%
\subsection{Additional Figures}
\label{app:figures}

\begin{figure}[t]
\centering
\begin{minipage}{0.78\linewidth}
  \begin{minipage}{0.49\linewidth}
    \begin{overpic}[width=\linewidth]{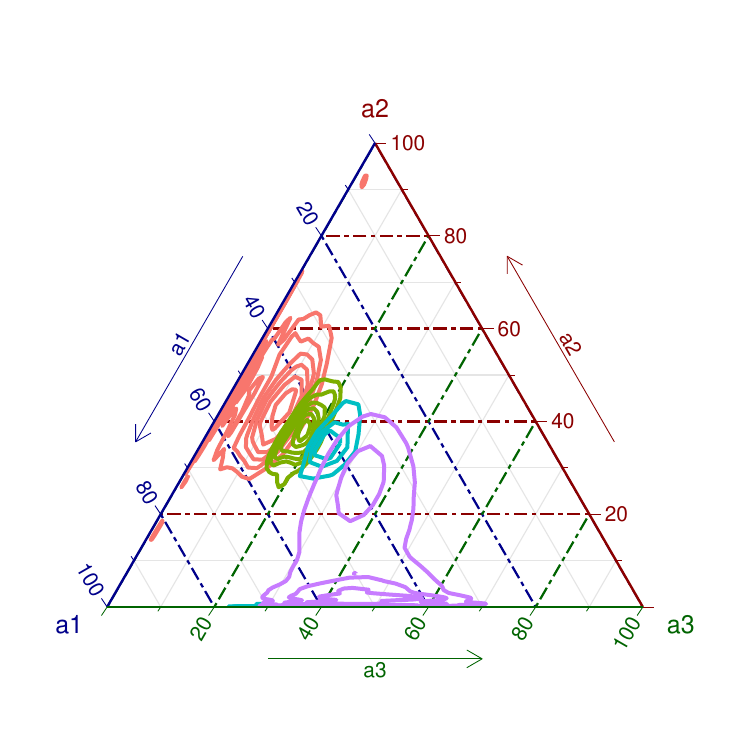}\put(-6,85){(a)}\end{overpic}
  \end{minipage}
  \begin{minipage}{0.49\linewidth}
    \begin{overpic}[width=\linewidth]{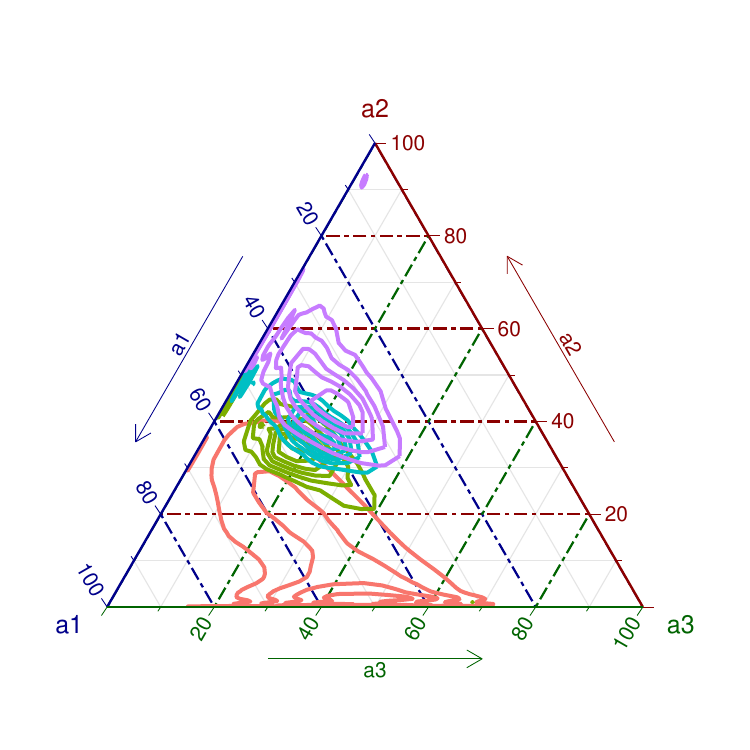}\put(-6,85){(b)}\end{overpic}
  \end{minipage}

  \begin{minipage}{0.49\linewidth}
    \begin{overpic}[width=\linewidth]{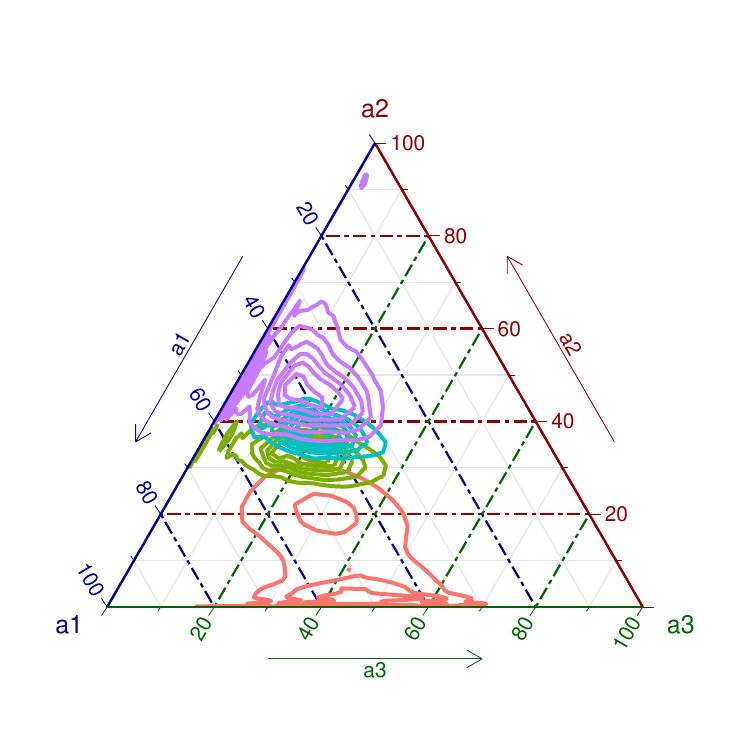}\put(-6,85){(c)}\end{overpic}
  \end{minipage}
  \begin{minipage}{0.49\linewidth}
    \begin{overpic}[width=\linewidth]{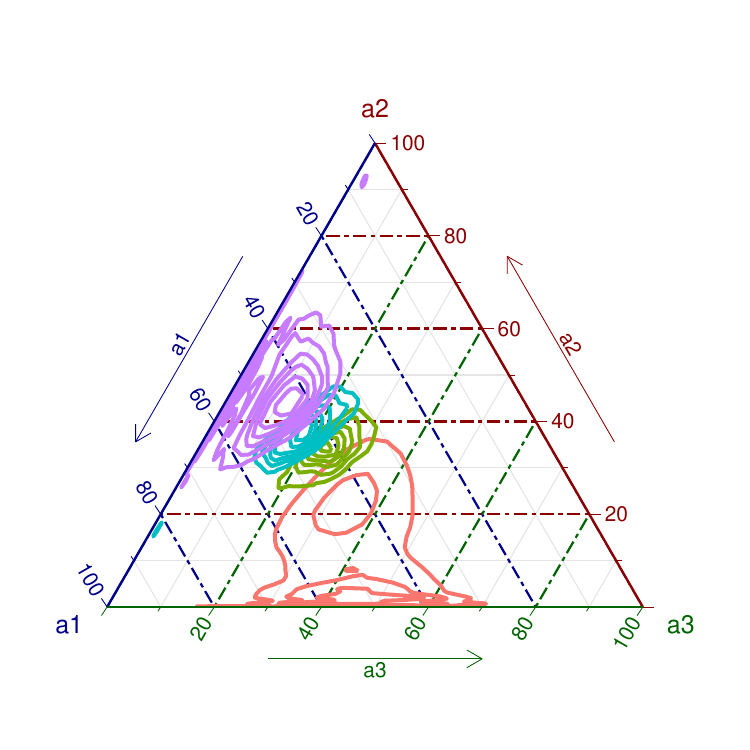}\put(-6,85){(d)}\end{overpic}
  \end{minipage}

  \begin{minipage}{0.49\linewidth}
    \begin{overpic}[width=\linewidth]{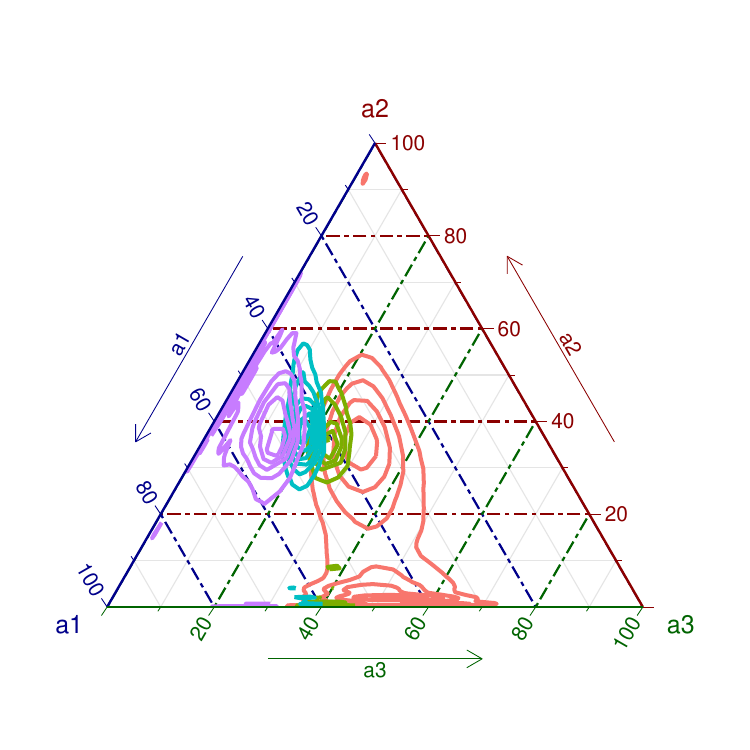}\put(-6,85){(e)}\end{overpic}
  \end{minipage}
  \begin{minipage}{0.49\linewidth}
    \begin{overpic}[width=\linewidth]{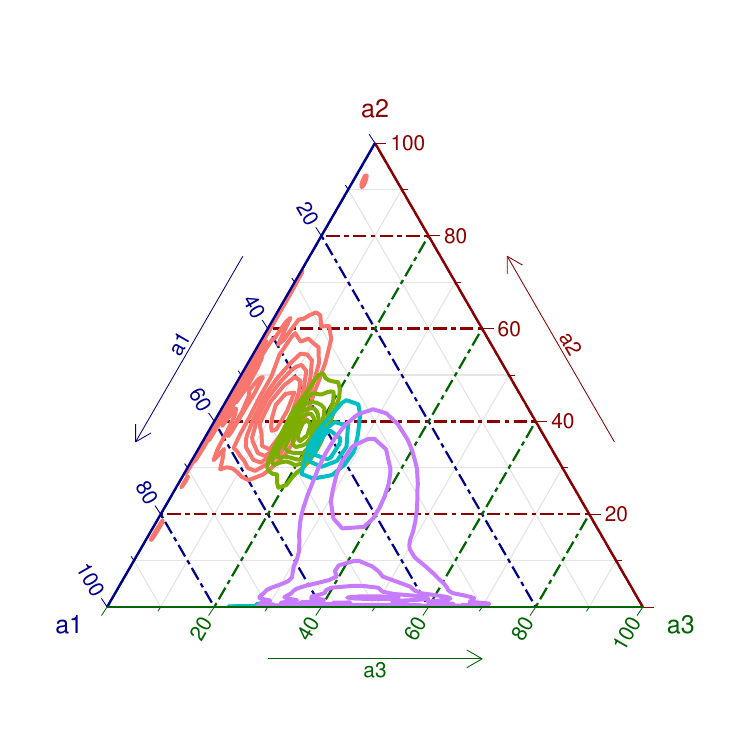}\put(-6,85){(f)}\end{overpic}
  \end{minipage}
\end{minipage}\hfill
\begin{minipage}{0.20\linewidth}
    \vspace*{0.15\linewidth}
    \includegraphics[width=\linewidth]{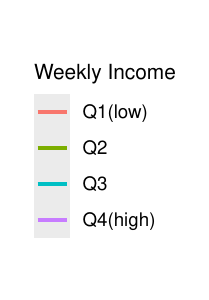}
\end{minipage}
\caption{Group-wise compositional distributions based on predicted outcomes: (a) NM, (b) Weighted, (c) Dirichlet, (d) LM, (e) KRR, and (g) KRR (Diri).}
\label{fig:app_prediction_plot}
\end{figure}

\begin{figure}[t]
\centering
\begin{minipage}{0.78\linewidth}
  \begin{minipage}{0.49\linewidth}
    \begin{overpic}[width=\linewidth]{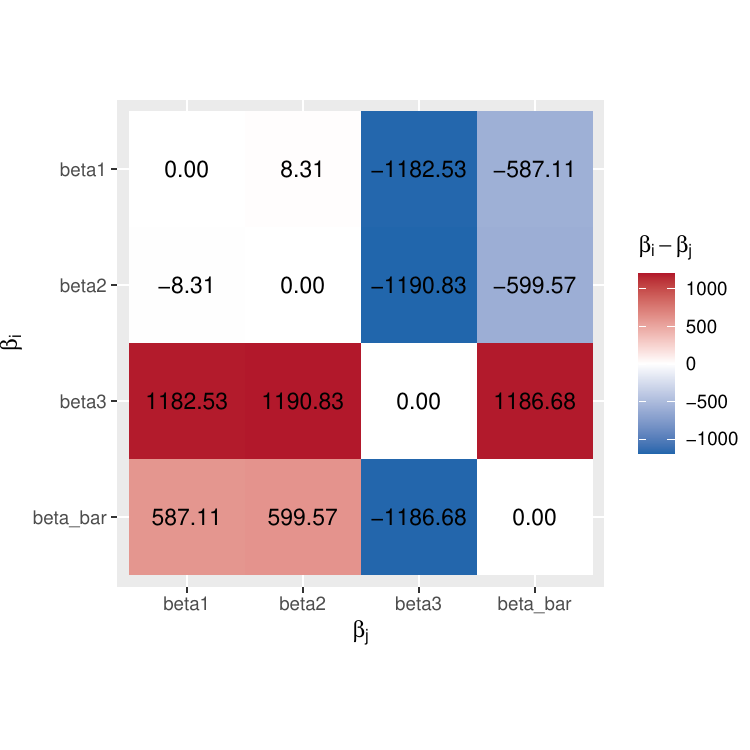}\put(-8,85){(a)}\end{overpic}
  \end{minipage}
  \hspace{0.05\linewidth}
  \begin{minipage}{0.49\linewidth}
    \begin{overpic}[width=\linewidth]{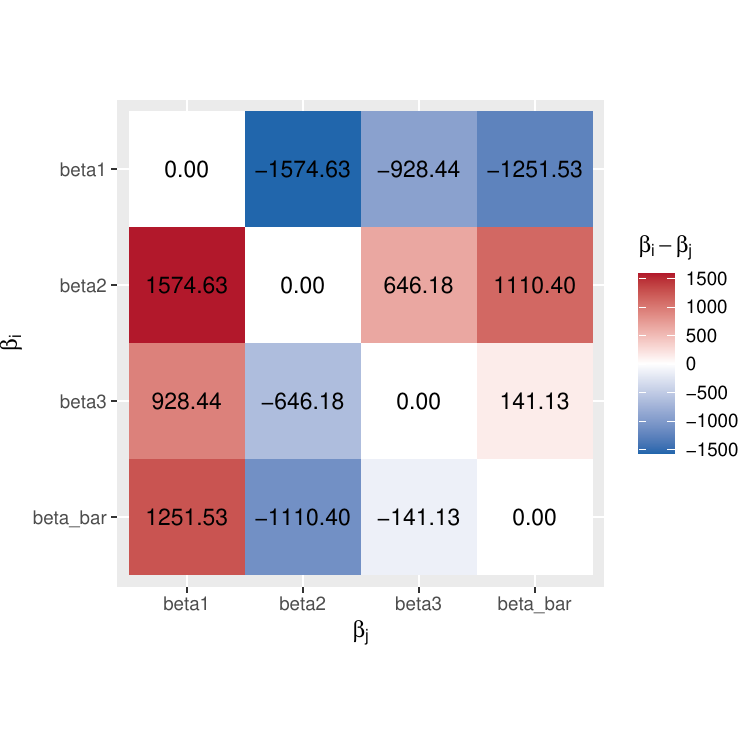}\put(-8,85){(b)}\end{overpic}
  \end{minipage}

  \begin{minipage}{0.49\linewidth}
    \begin{overpic}[width=\linewidth]{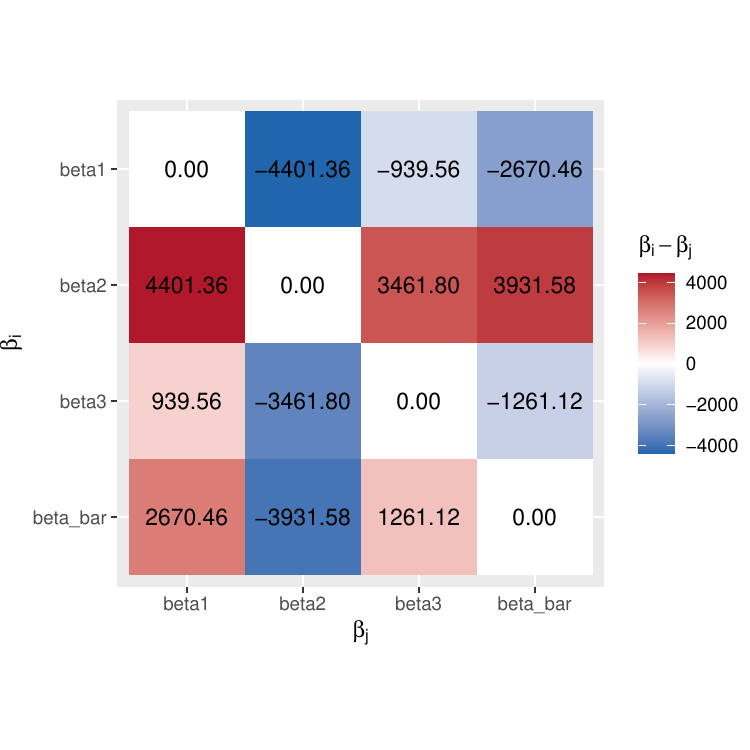}\put(-8,85){(c)}\end{overpic}
  \end{minipage}
  \hspace{0.05\linewidth}
  \begin{minipage}{0.49\linewidth}
    \begin{overpic}[width=\linewidth]{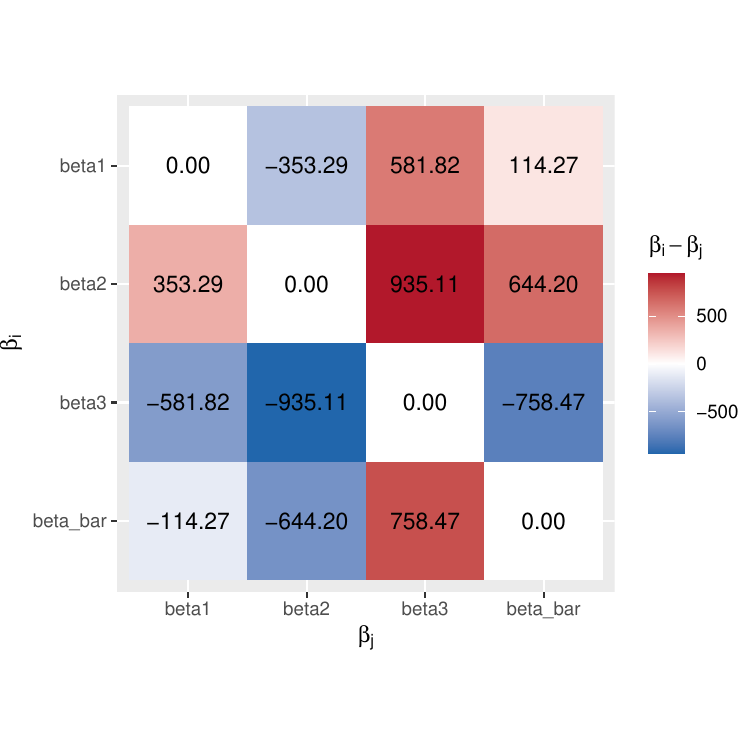}\put(-8,85){(d)}\end{overpic}
  \end{minipage}

  \begin{minipage}{0.49\linewidth}
    \begin{overpic}[width=\linewidth]{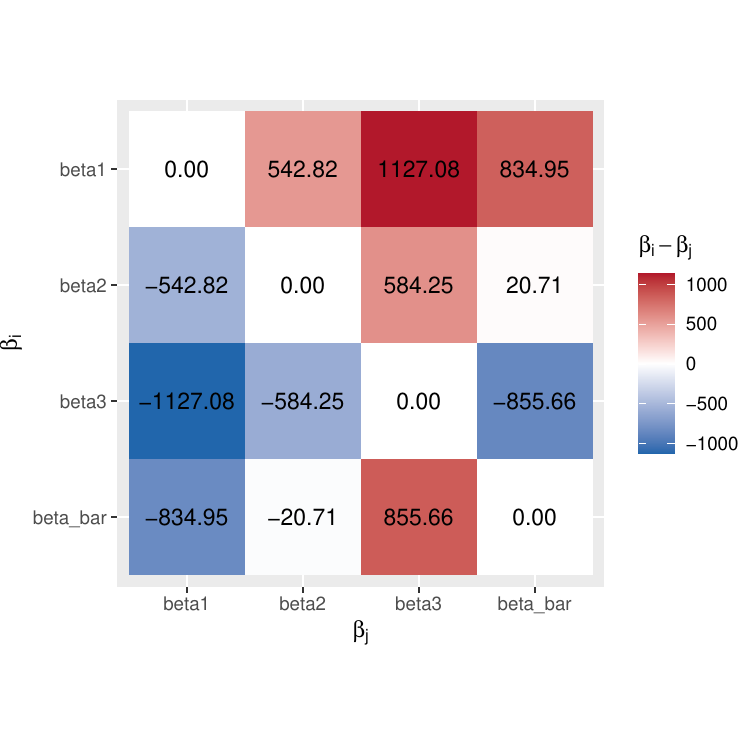}\put(-8,85){(e)}\end{overpic}
  \end{minipage}
  \hspace{0.05\linewidth}
  \begin{minipage}{0.49\linewidth}
    \begin{overpic}[width=\linewidth]{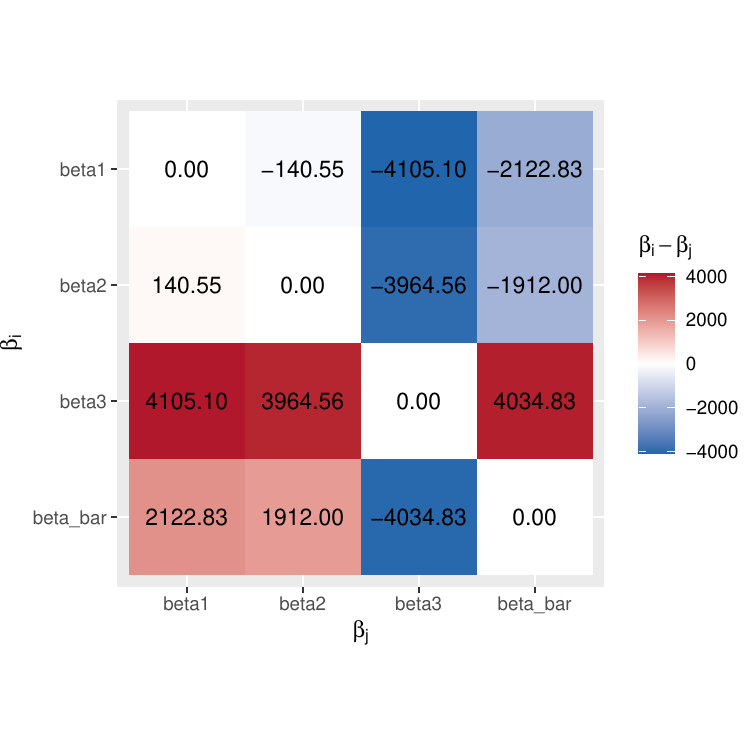}\put(-8,85){(f)}\end{overpic}
  \end{minipage}
\end{minipage}\hfill
\caption{Relative shift effects across estimators: (a) NM, (b) Weighted, (c) Dirichlet, (d) LM, (e) KRR, and (g) KRR (Diri). Each cell shows the change in weekly income for a one-hour reallocation $j$-th to $i$-th activity.}
\label{fig:app_reallocation_panels}
\end{figure}

\section{Proofs of Main Results}
%
%
\begin{theorem} [Representer and Decomposition Theorem]
\label{thm:representer_a}
    Let $\mathcal{H}$ be a reproducing kernel Hilbert space (RKHS) with product kernel
    $$ K\bigl((A,X),(A',X')\bigr) = K_A(A,A')\,K_X(X,X'). $$
    Let $\mathcal{H}(1)=\{f\in\mathcal{H}:\|f\|_{\mathcal{H}}\le1\}$ denote the unit ball of $\mathcal{H}$. 
    Then the worst-case balancing error satisfies
    $$ 
        \sup_{f\in\mathcal{H}(1)} S(w,f) = 
            \left\| (\mathbf{A}^\intercal \mathbf{A})^{-1} \mathbf{A}^\intercal 
            \bigl( \mathbf{W}, \, -\mathbf{I} \bigl)
            \tilde{\mathbf{K}}^{1/2} \right\|_{op}^2,
    $$
    where $\|\cdot\|_{op}$ denotes the operator norm, 
    $\mathbf{I}$ is the $n\times n$ identity matrix, 
    and $\bigl( \mathbf{W}, \, -\mathbf{I} \bigl)$ denotes the $n \times 2n$ block matrix.
    The matrix $\tilde{\mathbf{K}}$ is the $2n\times 2n$ kernel matrix defined in \eqref{eq:k_tilde_mat}.
\end{theorem}

\begin{proof}
    The functional $S(w,f)$ depends on $f$ only through the evaluations $\{ f(A_i,X_i) \}_{i=1}^n$ and the empirical averages $\{ \frac{1}{n} \sum_{k=1}^n f(A_i,X_k) \}_{i=1}^n$.
    By the reproducing property,
    $$
        f(A_i, X_i) = \langle f, \phi_{ii} \rangle, \quad 
        \frac{1}{n}\sum_{k=1}^n f(A_i, X_k) = \Bigl\langle f, \frac{1}{n}\sum_{k=1}^n \phi_{ik} \Bigl\rangle,
    $$
    where $\phi_{ik}=K(\cdot,(A_i,X_k))$.
    Hence $S(w,f)$ depends only on the projection of $f$ onto the finite-dimensional subspace
    $$
        V = \text{span}\left(\{\phi_{jj}\}_{j=1}^n \cup \left\{\frac{1}{n}\sum_{k=1}^n \phi_{jk}\right\}_{j=1}^n\right).
    $$
    Any component orthogonal to $V$ does not affect $S(w,f)$, so it suffices to consider $f\in V$, which admits the representation
    $$
        f = \sum_{j=1}^n \alpha_j \phi_{jj} + \sum_{j=1}^n \alpha'_j \left( \frac{1}{n} \sum_{k=1}^n \phi_{jk} \right).
    $$

    Using the product kernel structure, we obtain
    \begin{align*}
        f(A_i, X_i) &= 
            \sum_{j=1}^n \alpha_j \langle \phi_{jj}, \phi_{ii} \rangle + \sum_{j=1}^n \alpha'_j \frac{1}{n} \sum_{k=1}^n \langle \phi_{jk}, \phi_{ii} \rangle \\
        &= 
            \sum_{j=1}^n \alpha_j K_A(A_j, A_i) K_X(X_j, X_i) + \sum_{j=1}^n \alpha'_j K_A(A_j, A_i) \frac{1}{n} \sum_{k=1}^n K_X(X_k, X_i).
    \end{align*}
    and
    \begin{align*}
        \frac{1}{n} \sum_{k=1}^n f(A_i, X_k) &= 
            \sum_{j=1}^n \alpha_j \frac{1}{n} \sum_{k=1}^n \langle \phi_{jj}, \phi_{ik} \rangle + \sum_{j=1}^n \alpha'_j \frac{1}{n^2} \sum_{k=1}^n \sum_{l=1}^n \langle \phi_{jk}, \phi_{il} \rangle \\
        &= \sum_{j=1}^n \alpha_j K_A(A_j, A_i) \frac{1}{n} \sum_{k=1}^n K_X(X_j, X_k) + \sum_{j=1}^n \alpha'_j K_A(A_j, A_i) \frac{1}{n^2} \sum_{k=1}^n \sum_{l=1}^n K_X(X_k, X_l).
    \end{align*}

    Now let 
    $$
        \tilde{\mathbf{K}} = \begin{pmatrix} \mathbf{K}_1 & \mathbf{K}_2 \\ \mathbf{K}_2^\intercal & \mathbf{K}_3 \end{pmatrix},
        \quad
        \tilde{\boldsymbol{\alpha}} = \begin{pmatrix} \boldsymbol{\alpha} \\ \boldsymbol{\alpha}' \end{pmatrix},
    $$ 
    where
    \begin{align*}
        \mathbf{K}_1 &= \bigl( K_A(A_i, A_j) K_X(X_i, X_j) \bigr)_{ij}, \\
        \mathbf{K}_2 &= \bigl( \bar{K}_{X_j} K_A(A_i, A_j) \bigr)_{ij}, \quad
            \bar{K}_{X_j} = \frac{1}{n} \sum_{k=1}^n K_X(X_k, X_j), \\
        \mathbf{K}_3 &= \bar{K}_X \bigl( K_A(A_i, A_j) \bigr)_{ij}, \quad
            \bar{K}_X = \frac{1}{n^2} \sum_{k=1}^n \sum_{l=1}^n K_X(X_k, X_l).
    \end{align*}
    Then the two evaluation vectors can be written as
    \begin{align*}
        \bigl( f(A_1, X_1), \dots, f(A_n, X_n) \bigr)^\intercal = \mathbf{K}_1 \boldsymbol{\alpha} + \mathbf{K}_2 \boldsymbol{\alpha}', 
        \quad 
        \left( \frac{1}{n} \sum_{k=1}^n f(A_1, X_k), \dots, \frac{1}{n} \sum_{k=1}^n f(A_n, X_k) \right)^\intercal = \mathbf{K}_2^\intercal \boldsymbol{\alpha} + \mathbf{K}_3 \boldsymbol{\alpha}'.
    \end{align*}

    Substituting these expressions into $S(w,f)$ gives
    \begin{align*}
        S(w, f) &= 
            \left\| (\mathbf{A}^\intercal \mathbf{A})^{-1} \mathbf{A}^\intercal \left[ 
            \mathbf{W}(\mathbf{K}_1 \boldsymbol{\alpha} + \mathbf{K}_2 \boldsymbol{\alpha}')
            -
            (\mathbf{K}_2^\intercal \boldsymbol{\alpha} + \mathbf{K}_3 \boldsymbol{\alpha}') 
            \right] 
            \right\|_2^2 \\
        &= 
            \left\| (\mathbf{A}^\intercal \mathbf{A})^{-1} \mathbf{A}^\intercal \left[ (\mathbf{W}\mathbf{K}_1 - \mathbf{K}_2^\intercal)\boldsymbol{\alpha} + (\mathbf{W}\mathbf{K}_2 - \mathbf{K}_3)\boldsymbol{\alpha}' \right] \right\|_2^2 \\
        &=
            \left\| (\mathbf{A}^\intercal \mathbf{A})^{-1} \mathbf{A}^\intercal 
            \bigl( \mathbf{W}, \, -\mathbf{I} \bigl)
            \tilde{\mathbf{K}} \tilde{\boldsymbol{\alpha}} \right\|_2^2.
    \end{align*}

    Finally, the constraint $f\in\mathcal{H}(1)$ is equivalent to
    $$
        \|f\|_{\mathcal{H}}^2 = \tilde{\boldsymbol{\alpha}}^\intercal \tilde{\mathbf{K}} \tilde{\boldsymbol{\alpha}} \le 1.
    $$
    Therefore,
    $$ 
        \sup_{f\in\mathcal{H}(1)} S(w,f) = 
            \sup_{\tilde{\boldsymbol{\alpha}}^\intercal \tilde{\mathbf{K}}\tilde{\boldsymbol{\alpha}}\le 1} \left\| (\mathbf{A}^\intercal\mathbf{A})^{-1} \mathbf{A}^\intercal (\mathbf{W},-\mathbf{I}) \tilde{\mathbf{K}} \tilde{\boldsymbol{\alpha}} \right\|_2^2
        =
            \left\| (\mathbf{A}^\intercal \mathbf{A})^{-1} \mathbf{A}^\intercal 
            \bigl( \mathbf{W}, \, -\mathbf{I} \bigl)
            \tilde{\mathbf{K}}^{1/2} \right\|_{op}^2.
    $$
    This completes the proof.
    
\end{proof}
%
%
\begin{theorem}[Convexity]
\label{thm:convexity_a}
Let $\tilde{\mathbf{K}} \in \mathbb{R}^{2n \times 2n}$ admit a rank-$r$ approximation of the form
$$ 
    \tilde{\mathbf{K}} \approx \mathbf{P}_1 \mathbf{Q}_1 \mathbf{P}_1^\intercal,
$$
where $\mathbf{P}_1 \in \mathbb{R}^{2n \times r}$ has orthonormal columns and 
$\mathbf{Q}_1 = \mathrm{diag}(q_1,\dots,q_r)$ with $q_j \ge 0$. 
Let $\lambda \ge 0$, and define $\mathbf{W} = \mathrm{diag}(w)$ for $w \in \mathbb{R}^n$. Consider the objective function $\mathcal{L} : \mathbb{R}^n \to \mathbb{R}$ given by
$$ 
    \mathcal{L}(w) =
        \left\| (\mathbf{A}^\intercal \mathbf{A})^{-1}\mathbf{A}^\intercal
        (\mathbf{W}, -\mathbf{I})
        \mathbf{P}_1 \mathbf{Q}_1^{1/2}
        \right\|_{op}^2 
        +
        \lambda \left\| (\mathbf{A}^\intercal \mathbf{A})^{-1}\mathbf{A}^\intercal \mathbf{W} \right\|_F^2.
$$
Then $\mathcal{L}(w)$ is convex in $w$.
\end{theorem}

\begin{proof}
Since $\mathbf{W}=\mathrm{diag}(w)$ depends linearly on $w$, both
\[
    (\mathbf{A}^\intercal \mathbf{A})^{-1}\mathbf{A}^\intercal
    (\mathbf{W}, -\mathbf{I})
    \mathbf{P}_1 \mathbf{Q}_1^{1/2}
    \quad\text{and}\quad
    (\mathbf{A}^\intercal \mathbf{A})^{-1}\mathbf{A}^\intercal \mathbf{W}
\]
are affine functions of $w$. Because $\|\cdot\|_{op}^2$ and $\|\cdot\|_F^2$ are convex, both terms in $\mathcal{L}(w)$ are convex in $w$. Therefore, $\mathcal{L}(w)$ is convex.
\end{proof}
%
%
\begin{theorem}[Convergence rate of the weighted estimator]
\label{thm:conv_rate_a}
Suppose Assumptions~\ref{ass:Sigma_A}--\ref{ass:entropy} hold. Let
$$
    \hat{\boldsymbol{\beta}}
    =
    (\mathbf{A}^\intercal \mathbf{A})^{-1} \mathbf{A}^\intercal \hat{\mathbf{W}} \mathbf{Y},
$$
where $\hat{w}$ is defined in \eqref{eq:w_hat_def}, and define
$$
    \boldsymbol{\beta}^*
    =
    \boldsymbol{\Sigma}_A^{-1}
    \mathbb{E}\!\left[
        A \, \mathbb{E}_X\bigl[m(A,X)\bigr]
    \right].
$$
Then
$$
    \left\| \hat{\boldsymbol{\beta}} - \boldsymbol{\beta}^* \right\|_2
    =
    O_p \left( \frac{1}{\sqrt{n}} \big( \| m \|_{\mathcal{H}} C_1 + \| m \|_\infty C_2 + C_3 \big) \right),
$$
where $C_1$, $C_2$, and $C_3$ are constants determined by constants introduced in Assumptions~\ref{ass:Sigma_A}--\ref{ass:entropy} and the regularization parameter $\lambda \asymp 1$. See \eqref{eq:thm3_C_1}, \eqref{eq:thm3_C_2}, and \eqref{eq:thm3_C_3} for their explicit expressions.
\end{theorem}

\begin{proof}
    We decompose $\hat{\boldsymbol{\beta}} - \boldsymbol{\beta}^*$ into four terms as in \eqref{eq:thm3_part1}--\eqref{eq:thm3_part4} and bound each term separately.
    \begin{subequations} 
    \begin{align} 
        \hat{\boldsymbol{\beta}} - \boldsymbol{\beta}^* = & \,
        \hat{\boldsymbol{\Sigma}}_A^{-1} \left(
        \frac{1}{n} \sum_{i=1}^n A_i \left( \hat{w}_i m(A_i, X_i) - \frac{1}{n} \sum_{j=1}^n m(A_i, X_j) \right)
        \right) 
        \label{eq:thm3_part1} \\
        & +
        \hat{\boldsymbol{\Sigma}}_A^{-1} \left(
        \frac{1}{n} \sum_{i=1}^n A_i \left( \frac{1}{n} \sum_{j=1}^n m(A_i, X_j) - \mathbb{E}_X [m(A_i, X)] \right)
        \right) 
        \label{eq:thm3_part2} \\
        & +
        \hat{\boldsymbol{\Sigma}}_A^{-1} \left( \frac{1}{n} \sum_{i=1}^n A_i \hat{w}_i \varepsilon_i \right)
        \label{eq:thm3_part3} \\
        & +
        \hat{\boldsymbol{\Sigma}}_A^{-1} \left( \frac{1}{n} \sum_{i=1}^n A_i \left( \mathbb{E}_X [m(A_i, X)] - A_i^\intercal \boldsymbol{\beta}^* \right)\right)
        \label{eq:thm3_part4}
    \end{align} 
    \end{subequations} 

    The $\ell_2$ norm of the term in \eqref{eq:thm3_part1} can be bounded as follows.
    \begin{align*}
        & \left\| \hat{\boldsymbol{\Sigma}}_A^{-1} \left(
        \frac{1}{n} \sum_{i=1}^n A_i \left( \hat{w}_i m(A_i, X_i) - \frac{1}{n} \sum_{j=1}^n m(A_i, X_j) \right) 
        \right) \right\|_2 \\
        & \le
        \| m \|_{\mathcal{H}} \sup_{f \in \mathcal{H}(1)} \left\| \hat{\boldsymbol{\Sigma}}_A^{-1} \left(
        \frac{1}{n} \sum_{i=1}^n A_i \left( \hat{w}_i f(A_i, X_i) - \frac{1}{n} \sum_{j=1}^n f(A_i, X_j) \right)
        \right) \right\|_2 \\
        & =
            \| m \|_{\mathcal{H}} \sup_{f \in \mathcal{H}(1)} \sqrt{ S(\hat{w}, f) }
    \end{align*}
    Hence, by Lemma~\ref{lem:S(w_hat,f)_p(w_hat)_rate},
    $$
        \| m \|_{\mathcal{H}} \sup_{f \in \mathcal{H}(1)} \sqrt{ S(\hat{w}, f) }
        = 
        O_p\left( 
        \frac{\| m \|_{\mathcal{H}} \sqrt{L}}{\sqrt{c_A n}} \left( \sqrt{B}(C_w + 1) (\kappa_A \kappa_X)^{1-\alpha/2} + \sqrt{\lambda} C_w \right)
        \right).
    $$

    Now, for \eqref{eq:thm3_part2}, by Lemma~\ref{lem:cross_term_rate}, we have
    $$
        \left\| \hat{\boldsymbol{\Sigma}}_A^{-1} \left(
        \frac{1}{n} \sum_{i=1}^n A_i \left( \frac{1}{n} \sum_{j=1}^n m(A_i, X_j) - \mathbb{E}_X [m(A_i, X)] \right)
        \right) \right\|_2
        =
        O_p\!\left( \frac{\| m \|_\infty }{c_A \sqrt{n}} \right).
    $$
    
    For \eqref{eq:thm3_part3}, by Lemma~\ref{lem:error_term_bound},
    $$
        \left\|
        \hat{\boldsymbol{\Sigma}}_A^{-1} \left( \frac{1}{n} \sum_{i=1}^n A_i \hat{w}_i \varepsilon_i \right)
        \right\|_2
        =
        O_p\!\left( \frac{\sigma \sqrt{L} }{\sqrt{c_A n}} \left( \frac{ \sqrt{B} (C_w + 1) (\kappa_A \kappa_X)^{1 - \alpha/2}}{\sqrt{\lambda}} + C_w \right) \right).
    $$

    Last, for \eqref{eq:thm3_part4}, by Lemma~\ref{lem:m-beta*},
    $$
        \left\|
        \hat{\boldsymbol{\Sigma}}_A^{-1} \left( \frac{1}{n} \sum_{i=1}^n A_i \left( \mathbb{E}_X [m(A_i, X)] - A_i^\intercal \boldsymbol{\beta}^* \right)\right)
        \right\|_2 
        =
        O_p\!\left( \frac{\| m \|_\infty \sqrt{1 + c_A^{-2}} }{c_A \sqrt{n}} \right).
    $$

    Combining the results, we obtain
    $$
        \left\| \hat{\boldsymbol{\beta}} - \boldsymbol{\beta}^* \right\|_2 = O_p \left( \frac{1}{\sqrt{n}} \big( \| m \|_{\mathcal{H}} C_1 + \| m \|_\infty C_2 + C_3 \big) \right),
    $$
    where
    \begin{align}
        C_1 = 
        \frac{\sqrt{L B} (C_w + 1) (\kappa_A \kappa_X)^{1 - \alpha/2} + \sqrt{\lambda L} C_w}{\sqrt{c_A}},
        \label{eq:thm3_C_1}
    \end{align}
    \begin{align}
        C_2 =
        \frac{1 + \sqrt{1 + c_A^{-2}}}{c_A},
        \label{eq:thm3_C_2}
    \end{align}
    and
    \begin{align}
        C_3 =
        \frac{ \sigma \sqrt{L B} (C_w + 1) (\kappa_A \kappa_X)^{1 - \alpha/2}}{\sqrt{\lambda c_A}}
        +
        \frac{ \sigma \sqrt{L} C_w}{\sqrt{c_A}}.
        \label{eq:thm3_C_3}
    \end{align}
\end{proof}

%
%
\begin{theorem}[Convergence Rate of the Augmented Weighted Estimator] \label{thm:awe_rate_a}
    Under the conditions of Theorem~\ref{thm:conv_rate} and  Assumption~\ref{ass:delta_conv}, let $\hat{\boldsymbol{\beta}}_{\mathrm{AWE}}$ denote the Augmented Weighted Estimator defined in \eqref{eq:awe_estimator}. 
    If the regularization parameter satisfies $\lambda \asymp 1$, then
    $$
        \left\| \hat{\boldsymbol{\beta}}_{\text{AWE}} - \boldsymbol{\beta}^* \right\|_2 = O_p \left( \frac{1}{\sqrt{n}} \big( n^{-\zeta} C_1 + \| m \|_\infty C_2 + C_3 \big) \right).
    $$
    where $C_1$, $C_2$, and $C_3$ are defined in \eqref{eq:thm3_C_1}, \eqref{eq:thm3_C_2}, and \eqref{eq:thm3_C_3}, respectively.
\end{theorem}

\begin{proof}
    By the definition of $\tilde{Y}_i$, we can rewrite
    $$
        \tilde{Y}_i = \hat{w}_i\Bigl(m(A_i,X_i) + \varepsilon_i - \hat{m}(A_i,X_i)\Bigr) + \frac{1}{n} \sum_{j=1}^n \hat{m}(A_i,X_j).
    $$
    Then
    \begin{align*}
        \hat{\boldsymbol{\beta}}_{AWE} - \boldsymbol{\beta}^*  
        = & \,
        \hat{\boldsymbol{\Sigma}}_A^{-1} \left( \frac{1}{n}
        \sum_{i=1}^n A_i \left( \hat{w}_i \Delta(A_i, X_i) - \frac{1}{n} \sum_{j=1}^n \Delta(A_i, X_j) \right)
        \right) \\
        & +
        \hat{\boldsymbol{\Sigma}}_A^{-1} \left( \frac{1}{n}
        \sum_{i=1}^n A_i \left(
        \frac{1}{n} \sum_{j=1}^n m(A_i,X_j) - \mathbb{E}_X [m(A_i, X)]
        \right) \right) \\
        & +
        \hat{\boldsymbol{\Sigma}}_A^{-1} \left( \frac{1}{n}
        \sum_{i=1}^n A_i \hat{w}_i \varepsilon_i
        \right) \\
        & +
        \hat{\boldsymbol{\Sigma}}_A^{-1} \left( \frac{1}{n}
        \sum_{i=1}^n A_i \left(
        \mathbb{E}_X [m(A_i, X)] - A_i^\intercal \boldsymbol{\beta}^*
        \right) \right)
    \end{align*}
    where $\Delta (a, x) = m(a, x) - \hat{m} (a, x)$.
    This is exactly same as \eqref{eq:thm3_part2} -- \eqref{eq:thm3_part4} but \eqref{eq:thm3_part1} with inserting $\Delta(a, x)$ instead of $m(a, x)$. Therefore we can directly bring up the result of Theorem~\ref{thm:conv_rate_a} and plug in $\| \Delta \|_{\mathcal{H}}$ where the $ \| m \|_{\mathcal{H}}$ was, then 
    $$
        \left\| \hat{\boldsymbol{\beta}}_{AWE} - \boldsymbol{\beta}^* \right\|_2 = O_p \left( \frac{1}{\sqrt{n}} \big( \| \Delta \|_{\mathcal{H}} C_1 + \| m \|_\infty C_2 + C_3 \big) \right),
    $$
    where $C_1$, $C_2$, and $C_3$ are \eqref{eq:thm3_C_1}, \eqref{eq:thm3_C_2}, and \eqref{eq:thm3_C_3}, respectively.
    And then by Assumption~\ref{ass:delta_conv}, for $\zeta > 0$,
    $$
        \left\| \hat{\boldsymbol{\beta}}_{AWE} - \boldsymbol{\beta}^* \right\|_2 = O_p \left( \frac{1}{\sqrt{n}} \big( n^{-\zeta} C_1 + \| m \|_\infty C_2 + C_3 \big) \right).
    $$
\end{proof}
%
%
\begin{theorem}[Conditional Asymptotic Normality]
    Define
    $$
        \tilde{\boldsymbol{\beta}}
        =
        (\mathbf{A}^\intercal \mathbf{A})^{-1}
        \mathbf{A}^\intercal
        \hat{\mathbb{E}}_X[\mathbf{m}]
    \quad
    \text{where}
    \quad
        \hat{\mathbb{E}}_X[\mathbf{m}]
        :=
        \begin{pmatrix}
            \frac1n \sum_{j=1}^n m(A_1, X_j)
            \\
            \vdots
            \\
            \frac1n \sum_{j=1}^n m(A_n, X_j)
        \end{pmatrix},
    $$
    and suppose that $ Y_i = m(A_i, X_i) + \varepsilon_i $.
    Under the conditions of Theorem~\ref{thm:awe_rate} and  Assumption~\ref{ass:eps_third_moment} and \ref{ass:weighted_stability}, it holds that
    $$
        \sqrt{n} 
        \left\{ 
        \hat{\boldsymbol{\Sigma}}_A^{-1} \left( \frac{1}{n} \sum_{i=1}^n \hat{w}_i^2 \hat{\varepsilon}_i^2 A_i A_i^\intercal \right) \hat{\boldsymbol{\Sigma}}_A^{-1} 
        \right\}^{-1/2} 
        \left( 
        \hat{\boldsymbol{\beta}}_{AWE} - \tilde{\boldsymbol{\beta}} 
        \right) 
        \overset{d}{\longrightarrow} \mathcal{N}(\mathbf{0}, \mathbf{I} ).
    $$
    where $ \hat{\varepsilon}_i = Y_i - \hat{m} (A_i, X_i) $.
\end{theorem}

\begin{proof}
    $$
        \hat{\boldsymbol{\beta}}_{AWE} - \tilde{\boldsymbol{\beta}} 
        =
        \left( \frac{1}{n} \mathbf{A}^\intercal \mathbf{A} \right)^{-1} 
        \left\{ \frac{1}{n} \sum_{i=1}^n A_i \left( \hat{w}_i \varepsilon_i + \hat{w}_i \Delta(A_i, X_i) - \frac1n \sum_{j=1}^n \Delta(A_i, X_j) \right) \right\}
    $$
    where $\Delta = m - \hat{m}$. Then
    \begin{align*}
        \sqrt{n} \Bigl( \hat{\boldsymbol{\beta}}_{AWE} - \tilde{\boldsymbol{\beta}} \Bigl) =&
            \hat{\boldsymbol{\Sigma}}_A^{-1} \left\{ \frac{1}{\sqrt{n}} \sum_{i=1}^n A_i \left( \hat{w}_i \varepsilon_i + \hat{w}_i \Delta(A_i, X_i) - \frac1n \sum_{j=1}^n \Delta(A_i, X_j) \right) \right\} \\
        =&
           \hat{\boldsymbol{\Sigma}}_A^{-1} \left\{ \frac{1}{\sqrt{n}} \sum_{i=1}^n \hat{w}_i \varepsilon_i A_i 
           + 
           \frac{1}{\sqrt{n}} \sum_{i=1}^n A_i \left( \hat{w}_i \Delta(A_i, X_i) - \frac1n \sum_{j=1}^n \Delta(A_i, X_j) \right) \right\} \\
        =&
            \hat{\boldsymbol{\Sigma}}_A^{-1} 
            \left( \frac{1}{\sqrt{n}} \sum_{i=1}^n \hat{w}_i \varepsilon_i A_i \right) 
            + 
            \sqrt{n} \, \hat{\boldsymbol{\Sigma}}_A^{-1}
            \left\{
            \frac{1}{n} \sum_{i=1}^n A_i \left( \hat{w}_i \Delta(A_i, X_i) - \frac1n \sum_{j=1}^n \Delta(A_i, X_j) \right)
            \right\}
    \end{align*}
    In Theorem~\ref{thm:conv_rate_a}, when we calculated the convergence rate of \eqref{eq:thm3_part1}, we got 
    \begin{align*}
        \left\| \hat{\boldsymbol{\Sigma}}_A^{-1} \left\{
        \frac{1}{n} \sum_{i=1}^n A_i \left( \hat{w}_i m(A_i, X_i) - \frac{1}{n} \sum_{j=1}^n m(A_i, X_j) \right) 
        \right\} \right\|_2 
        & \le 
        \| m \|_{\mathcal{H}} \sup_{f \in \mathcal{H}(1)} \sqrt{ S(\hat{w}, f) } \\
        & =
        O_p\left( \frac{1}{\sqrt{n}} \| m \|_{\mathcal{H}} C_1 \right)
    \end{align*}
    where $C_1$ is \eqref{eq:thm3_C_1}. Applying the same argument with $\Delta$ in place of $m$, we obtain
    $$
        \sqrt{n} \, \left\| \hat{\boldsymbol{\Sigma}}_A^{-1}
        \left\{
        \frac{1}{n} \sum_{i=1}^n A_i \left( \hat{w}_i \Delta(A_i, X_i) - \frac1n \sum_{j=1}^n \Delta(A_i, X_j) \right)
        \right\} \right\|_2
        = 
        \sqrt{n} \, O_p\left( \frac{1}{\sqrt{n}} n^{-\zeta} C_1 \right) 
        = 
        O_p\left( n^{-\zeta} C_1 \right) 
    $$
    which is $o_p(1)$. Hence, the second term is asymptotically negligible. 
    
    Let 
    $$
        T_n := \hat{\boldsymbol{\Sigma}}_A^{-1} \left[ \frac{1}{\sqrt{n}} \sum_{i=1}^n \hat{w}_i \varepsilon_i A_i \right].
    $$
    Then Lemma~\ref{lem:normality} shows that
    $$
        \mathbf{V}_n^{-1/2} T_n \overset{d}{\longrightarrow} \mathcal{N}(\mathbf{0}, \mathbf{I} ),
    $$
    where
    $$
        \mathbf{V}_n
        :=
        \hat{\boldsymbol{\Sigma}}_A^{-1} \left( \frac1n \sum_{i=1}^n \hat w_i^2 \sigma_i^2 A_i A_i^\intercal \right) \hat{\boldsymbol{\Sigma}}_A^{-1}.
    $$

    Define $ \hat{\varepsilon}_i = Y_i - \hat{m}(A_i, X_i). $
    Since $Y_i = m(A_i,X_i)+\varepsilon_i$,
    $$ 
        \hat{\varepsilon}_i - \varepsilon_i = \Delta_i,
        \qquad
        \Delta_i = m(A_i,X_i) - \hat{m}(A_i,X_i).
    $$
    By Assumption~\ref{ass:delta_conv} and \ref{ass:kernels},
    $$
        \max_{1 \le i \le n} |\Delta_i| = O_p(n^{-\zeta}) = o_p(1).
    $$

    Let
    $$
        \hat{\mathbf{V}}_n
        :=
        \hat{\boldsymbol{\Sigma}}_A^{-1}
        \left(
        \frac1n \sum_{i=1}^n
        \hat w_i^2 \hat{\varepsilon}_i^2 A_i A_i^\intercal
        \right)
        \hat{\boldsymbol{\Sigma}}_A^{-1}.
    $$
    Then
    $$
        \hat{\varepsilon}_i^2 - \sigma_i^2
        =
        (\varepsilon_i^2 - \sigma_i^2) + 2 \varepsilon_i \Delta_i + \Delta_i^2.
    $$
    Hence
    \begin{align*}
        \hat{\mathbf{V}}_n - \mathbf{V}_n
        =
        \hat{\boldsymbol{\Sigma}}_A^{-1}
        \left[
        \frac1n \sum_{i=1}^n
        \hat w_i^2
        (\hat{\varepsilon}_i^2 - \sigma_i^2)
        A_i A_i^\intercal
        \right]
        \hat{\boldsymbol{\Sigma}}_A^{-1}.
    \end{align*}

    The term involving $(\varepsilon_i^2 - \sigma_i^2)$ is $o_p(1)$ under the same moment and weight conditions as in Lemma~\ref{lem:normality}. The remaining terms are also $o_p(1)$ since $\max_i |\Delta_i| = o_p(1)$, $\|A_i\|_2 \le 1$, and $n^{-1} \sum_{i=1}^n \hat w_i^2 = O_p(1)$. Therefore,
    $$
        \|\hat{\mathbf{V}}_n - \mathbf{V}_n\|_{op} = o_p(1).
    $$
    Since $\phi_{\min}(\mathbf{V}_n)$ is bounded away from zero with probability tending to one by Assumption~\ref{ass:weighted_stability},
    $$
        \mathbf{V}_n^{-1/2} \hat{\mathbf{V}}_n \mathbf{V}_n^{-1/2}
        =
        \mathbf{I} + o_p(1).
    $$

    By Slutsky's theorem,
    $$
        \hat{\mathbf{V}}_n^{-1/2} T_n
        \overset{d}{\longrightarrow}
        \mathcal{N}(\mathbf{0}, \mathbf{I}).
    $$
    Since
    $$
        \sqrt{n}
        \left(
        \hat{\boldsymbol{\beta}}_{AWE} - \tilde{\boldsymbol{\beta}}
        \right)
        =
        T_n,
    $$
    it follows that
    $$
        \sqrt{n} 
        \left\{ 
        \hat{\boldsymbol{\Sigma}}_A^{-1}
        \left(
        \frac{1}{n} \sum_{i=1}^n
        \hat{w}_i^2 \hat{\varepsilon}_i^2 A_i A_i^\intercal
        \right)
        \hat{\boldsymbol{\Sigma}}_A^{-1}
        \right\}^{-1/2}
        \left(
        \hat{\boldsymbol{\beta}}_{AWE} - \tilde{\boldsymbol{\beta}}
        \right)
        \overset{d}{\longrightarrow}
        \mathcal{N}(\mathbf{0}, \mathbf{I}).
    $$
\end{proof}

\section{Auxiliary Lemmas}
%
%
\begin{lemma} \label{lem:sigmaA}
    Let $\{ A_i \}_{i=1}^n$ be independent and identically distributed random vectors taking values in the simplex
    $$
        \mathcal{A} := \left\{ a \in \mathbb{R}^L : a_l \ge 0,\; \sum_{l = 1}^L a_l = 1 \right\}.
    $$
    Let $\hat{\boldsymbol{\Sigma}}_A = \frac{1}{n} \sum_{i=1}^n A_i A_i^\intercal$ be the sample second moment matrix and $\boldsymbol{\Sigma}_A = \mathbb{E}[A A^\intercal ]$ its population counterpart. 
    Under Assumption~\ref{ass:Sigma_A}, we have
    $$
        \left\rVert \hat{\boldsymbol{\Sigma}}_A^{-1} \right\rVert_{op} = O_p\left( \frac{1}{c_A} \right).
    $$
\end{lemma}

\begin{proof}
    The random vectors $\{A_i\}_{i=1}^n$ are i.i.d. and uniformly bounded, since $\|A_i\|_2 \le 1$. Standard concentration results for sums of independent random matrices imply that
    $$
        \left\| \hat{\boldsymbol{\Sigma}}_A - \boldsymbol{\Sigma}_A \right\|_{op} = O_p\left(\sqrt{\frac{\log L}{n}}\right).
    $$
    
    By Weyl's inequality, the smallest eigenvalue satisfies
    $$
        \phi_{\min} \left( \hat{\boldsymbol{\Sigma}}_A \right) 
        \ge 
        \phi_{\min} \left( \boldsymbol{\Sigma}_A \right) 
        - 
        \left\| \hat{\boldsymbol{\Sigma}}_A - \boldsymbol{\Sigma}_A \right\|_{op}.
    $$

    Under Assumption~\ref{ass:Sigma_A}, we have $\phi_{\min}(\boldsymbol{\Sigma}_A) \ge c_A > 0$. 
    Hence,
    $$
        \phi_{\min}(\hat{\boldsymbol{\Sigma}}_A) 
        \ge 
        c_A - O_p\left(\sqrt{\frac{\log L}{n}}\right) 
        = c_A - o_p(1).
    $$

    Therefore, $\phi_{\min}(\hat{\boldsymbol{\Sigma}}_A)$ is bounded away from zero in probability for sufficiently large $n$. 
    Therefore, the operator norm of the inverse matrix is
    $$
        \left\rVert \hat{\boldsymbol{\Sigma}}_A^{-1} \right\rVert_{op} 
        = 
        \frac{1}{\phi_{\min}(\hat{\boldsymbol{\Sigma}}_A)} 
        = O_p\left( \frac{1}{c_A} \right).
    $$
\end{proof}

%
%
\begin{lemma} \label{lem:M_nl}
    Let $\{(A_i, X_i)\}_{i=1}^n$ be independent and identically distributed, where $A_i = (a_{i1}, \dots, a_{iL})^\intercal$ denotes the treatment vector and $w_i^* = w^*(A_i, X_i)$ denotes the oracle weight defined in \eqref{eq:w*definition}. 
    For any $l=1,\dots,L$, define the centered sample mean process as
    $$ 
        M_{n,l}(f) = \frac{1}{n}\sum_{i=1}^{n} \Bigl( \xi_{i,l}(f) - \mathbb{E}[\xi_{i,l}(f)] \Bigr),
    $$
    where $\xi_{i,l}(f) = a_{il}\bigl( w_{i}^{*}f(A_{i},X_{i}) - \mathbb{E}_{X}[f(A_{i},X)] \bigr)$ and $\mathbb{E}_X$ denotes expectation with respect to the marginal distribution of $X$. 
    Then, under Assumptions~\ref{ass:C_w}, \ref{ass:kernels}, and \ref{ass:entropy}, we have
    $$
        \sup_{f\in\mathcal{H}(1)}|M_{n,l}(f)| = O_{p} \left(\frac{\sqrt{B} (C_w + 1) (\kappa_A \kappa_X)^{1-\alpha/2}}{\sqrt{n}}\right).
    $$
\end{lemma}

\begin{proof}
Define the $\sqrt{n}$-normalized empirical process $Z_{n,l}(f) = \sqrt{n}M_{n,l}(f)$. For $f_1,f_2\in\mathcal{H}(1)$, let $\delta_f = f_1 - f_2$. By Assumption~\ref{ass:C_w}, we have
\begin{align*}
    |\xi_{i,l}(\delta_f)| &= 
        \left| a_{il}\bigl( w_i^* \delta_f(A_i,X_i) - \mathbb{E}_X[\delta_f(A_i,X)] \bigr) \right| \\
    &\le 
        |a_{il}|\Big( |w_i^* \delta_f(A_i,X_i)| + |\mathbb{E}_X[\delta_f(A_i,X)]| \Big) \\
    &\le 
        (C_w+1)\|\delta_f\|_\infty.
\end{align*}
By Assumption~\ref{ass:kernels}, $\|f\|_\infty \le \kappa_A\kappa_X$ for all $f\in\mathcal{H}(1)$, so that 
$$ |\xi_{i,l}(\delta_f)| \le 2 \kappa_A \kappa_X (C_w+1). $$

Define the pseudo-metric 
$$ d(f_1,f_2)^2 = \mathbb{E}\big[(\xi_{i,l}(f_1)-\xi_{i,l}(f_2))^2\big]. $$
The above bound yields 
$$ d(f_1,f_2) \le (C_w+1)\|f_1-f_2\|_\infty, $$
which implies, by Assumption~\ref{ass:entropy},
$$ 
    (\mathcal{H}(1), d, \epsilon) 
    \le 
    H\bigl(\mathcal{H}(1), \|\cdot\|_\infty, \epsilon/(C_w+1)\bigr) 
    \le 
    B(C_w+1)^\alpha \epsilon^{-\alpha}.
$$

Since $\{Z_{n,l}(f)\}_{f\in\mathcal{H}(1)}$ is a centered empirical process, Dudley's entropy integral bound yields
$$
    \mathbb{E} \left[ \sup_{f\in\mathcal{H}(1)} |Z_{n,l}(f)| \right]
    \le 
    K \int_0^D \sqrt{H(\mathcal{H}(1), d, \epsilon)}\, d\epsilon,
$$
where the diameter satisfies $D \le 2 \kappa_A \kappa_X (C_w+1)$. Consequently,
\begin{align*}
    \mathbb{E} \left[ \sup_{f\in\mathcal{H}(1)} |Z_{n,l}(f)| \right] &\le 
        K \int_0^D \sqrt{B(C_w+1)^\alpha \epsilon^{-\alpha}}\, d\epsilon \\
    &= 
        K\sqrt{B}(C_w+1)^{\alpha/2} \int_0^D \epsilon^{-\alpha/2} d\epsilon \\
    &\le 
        K' \sqrt{B} (C_w+1) (\kappa_A \kappa_X)^{1-\alpha/2},
\end{align*}
for a constant $K'$ independent of $n$. This shows that 
$$ 
    \mathbb{E} \left[ \sup_{f\in\mathcal{H}(1)} |Z_{n,l}(f)| \right] = O \left( \sqrt{B} (C_w + 1) (\kappa_A \kappa_X)^{1-\alpha/2} \right), 
$$
and hence $\sup_{f\in\mathcal{H}(1)} |Z_{n,l}(f)| = O_p \Big(\sqrt{B} (C_w + 1) (\kappa_A \kappa_X)^{1-\alpha/2} \Big)$. Therefore,
$$
    \sup_{f\in\mathcal{H}(1)} |M_{n,l}(f)| 
    = 
    \frac{1}{\sqrt{n}}\sup_{f\in\mathcal{H}(1)} |Z_{n,l}(f)| 
    = 
    O_p \! \left(\frac{\sqrt{B} (C_w + 1) (\kappa_A \kappa_X)^{1-\alpha/2}}{\sqrt{n}}\right).
$$
\end{proof}

%
%
\begin{lemma}
    \label{lem:S(w*,f)_part1}
    Let $\mathbf{f} = (f(A_1,X_1), \dots, f(A_n,X_n))^\intercal$.
    Define $\mathbf{S}_1(f)$ and its expectation $\boldsymbol{\mu}_1(f)$ by
    $$
        \mathbf{S}_1(f) = (\mathbf{A}^\intercal \mathbf{A})^{-1} \mathbf{A}^\intercal \left( \mathbf{W}^* \mathbf{f} - \mathbb{E}_{X}[\mathbf{f}] \right),
        \quad
        \boldsymbol{\mu}_1(f) = \mathbb{E}[\mathbf{S}_1(f)].
    $$
    Under the conditions of Lemma~\ref{lem:sigmaA} and Lemma~\ref{lem:M_nl}, it holds that
    $$
        \sup_{f \in \mathcal{H}(1)} 
        \left\| \mathbf{S}_1(f) - \boldsymbol{\mu}_1(f) \right\|_2^2 
        = O_p\!\left( \frac{L B (C_w + 1)^2 (\kappa_A \kappa_X)^{2-\alpha}}{c_A n} \right).
    $$
\end{lemma}

\begin{proof}
    Let $R_i(f) = A_i \bigl(w_i^* f(A_i, X_i) - \mathbb{E}_X[f(A_i, X_i)]\bigr)$. Then
    $$
        \mathbf{S}_1(f) 
        = \hat{\boldsymbol{\Sigma}}_A^{-1} \frac{1}{n} \sum_{i=1}^n R_i(f),
        \quad
        \boldsymbol{\mu}_1(f) 
        = \mathbb{E}\!\left[\hat{\boldsymbol{\Sigma}}_A^{-1} \frac{1}{n} \sum_{i=1}^n R_i(f)\right].
    $$
    Hence,
    $$
        \mathbf{S}_1(f) - \boldsymbol{\mu}_1(f)
        = \hat{\boldsymbol{\Sigma}}_A^{-1} \frac{1}{n} \sum_{i=1}^n \bigl(R_i(f) - \mathbb{E}[R_i(f)]\bigr).
    $$
    It follows that
    \begin{align*}
        \sup_{f \in \mathcal{H}(1)} 
        \left\| \mathbf{S}_1(f) - \boldsymbol{\mu}_1(f) \right\|_2^2
        &\le 
        \left\| \hat{\boldsymbol{\Sigma}}_A^{-1} \right\|_{op}^2 
        \cdot
        \sup_{f \in \mathcal{H}(1)} 
        \left\| \frac{1}{n} \sum_{i=1}^n \bigl(R_i(f) - \mathbb{E}[R_i(f)]\bigr) \right\|_2^2.
    \end{align*}
    By Lemma~\ref{lem:sigmaA}, $\|\hat{\boldsymbol{\Sigma}}_A^{-1}\|_{op}^2 = O_p(1 / c_A)$.

    For the second term, the $l$-th component of the vector
    $\frac{1}{n} \sum_{i=1}^n (R_i(f) - \mathbb{E}[R_i(f)])$
    coincides with the centered empirical process $M_{n,l}(f)$ defined in Lemma~\ref{lem:M_nl}. Hence,
    $$
        \left\| \frac{1}{n} \sum_{i=1}^n \bigl(R_i(f) - \mathbb{E}[R_i(f)]\bigr) \right\|_2^2
        = \sum_{l=1}^L \bigl(M_{n,l}(f)\bigr)^2.
    $$
    Therefore,
    $$
        \sup_{f \in \mathcal{H}(1)} \sum_{l=1}^L \bigl(M_{n,l}(f)\bigr)^2
        \le 
        \sum_{l=1}^L \sup_{f \in \mathcal{H}(1)} \bigl(M_{n,l}(f)\bigr)^2.
    $$
    By Lemma~\ref{lem:M_nl}, for each $l$,
    $$
        \sup_{f \in \mathcal{H}(1)} |M_{n,l}(f)|^2 
        = O_p\!\left( \frac{B (C_w + 1)^2 (\kappa_A \kappa_X)^{2-\alpha}}{n} \right).
    $$
    Summing over $l = 1,\dots,L$ yields
    $$
        \sum_{l=1}^L \sup_{f \in \mathcal{H}(1)} \bigl(M_{n,l}(f)\bigr)^2
        = O_p\!\left( \frac{L B (C_w + 1)^2 (\kappa_A \kappa_X)^{2-\alpha}}{n} \right).
    $$
    Combining the bounds completes the proof.
\end{proof}

%
%

\begin{lemma}
    \label{lem:S(w*,f)_part2}
    Under the conditions of Lemma~\ref{lem:sigmaA} and Lemma~\ref{lem:M_nl}, it holds that
    $$
        \sup_{f \in \mathcal{H}(1)} 
        \left\| (\mathbf{A}^\intercal \mathbf{A})^{-1} \mathbf{A}^\intercal 
        \left( \mathbb{E}_X[\mathbf{f}] - \hat{\mathbb{E}}_X [\mathbf{f}] \right) \right\|_2^2 
        = O_p\!\left( \frac{L B(\kappa_A \kappa_X)^{2-\alpha}}{c_A n} \right).
    $$
\end{lemma}

\begin{proof}
    We begin by applying the inequality $\|\mathbf{M}\mathbf{v}\|_2 \le \|\mathbf{M}\|_F \|\mathbf{v}\|_2$, which yields
    \begin{align*}
        &\sup_{f \in \mathcal{H}(1)} 
        \left\| (\mathbf{A}^\intercal \mathbf{A})^{-1} \mathbf{A}^\intercal 
        \left( \mathbb{E}_X[\mathbf{f}] - \hat{\mathbb{E}}_X [\mathbf{f}] \right) \right\|_2^2 \\
        &\qquad \le 
        \left\| (\mathbf{A}^\intercal \mathbf{A})^{-1} \mathbf{A}^\intercal \right\|_F^2 
        \sup_{f \in \mathcal{H}(1)} 
        \left\| \mathbb{E}_X[\mathbf{f}] - \hat{\mathbb{E}}_X [\mathbf{f}] \right\|_2^2.
    \end{align*}

    We first bound the matrix term. Using $\mathbf{A}^\intercal \mathbf{A} = n\hat{\boldsymbol{\Sigma}}_A$ and the identity
    $\|(\mathbf{A}^\intercal \mathbf{A})^{-1} \mathbf{A}^\intercal\|_F^2 = \operatorname{tr}((\mathbf{A}^\intercal \mathbf{A})^{-1})$, we obtain
    $$
        \left\| (\mathbf{A}^\intercal \mathbf{A})^{-1} \mathbf{A}^\intercal \right\|_F^2
        = \operatorname{tr}\!\left( \big( n\hat{\boldsymbol{\Sigma}}_A \big)^{-1} \right)
        = \frac{1}{n} \operatorname{tr}\!\left( \hat{\boldsymbol{\Sigma}}_A^{-1} \right).
    $$
    Since
    $
        \operatorname{tr}(\hat{\boldsymbol{\Sigma}}_A^{-1})
        \le L \left\|\hat{\boldsymbol{\Sigma}}_A^{-1} \right\|_{op},
    $
    Lemma~\ref{lem:sigmaA} implies
    $$
        \left\| (\mathbf{A}^\intercal \mathbf{A})^{-1} \mathbf{A}^\intercal \right\|_F^2
        = O_p\!\left( \frac{L}{c_A n} \right).
    $$

    Next, define
    $$
        \delta_i(f) = \mathbb{E}_X[f(A_i, X)] - \frac{1}{n}\sum_{j=1}^n f(A_i, X_j).
    $$
    Then
    $$
        \left\| \mathbb{E}_X[\mathbf{f}] - \hat{\mathbb{E}}_X [\mathbf{f}] \right\|_2^2
        = \sum_{i=1}^n \delta_i(f)^2.
    $$
    Hence,
    $$
        \sup_{f \in \mathcal{H}(1)} \sum_{i=1}^n \delta_i(f)^2
        \le \sum_{i=1}^n \sup_{f \in \mathcal{H}(1)} \delta_i(f)^2.
    $$
    For each $i$, the term $\delta_i(f)$ is a centered empirical mean over the class 
    $\{f(A_i,\cdot): f \in \mathcal{H}(1)\}$. By the same argument as in Lemma~\ref{lem:M_nl},
    $$
        \sup_{f \in \mathcal{H}(1)} |\delta_i(f)| 
        = O_p\!\left( \frac{\sqrt{B} (\kappa_A \kappa_X)^{1 - \alpha/2}}{\sqrt{n}} \right).
    $$
    Therefore,
    $$
        \sum_{i=1}^n \sup_{f \in \mathcal{H}(1)} \delta_i(f)^2
        = O_p \big(B (\kappa_A \kappa_X)^{2-\alpha} \big).
    $$

    Combining the bounds yields
    $$
        \sup_{f \in \mathcal{H}(1)} 
        \left\| (\mathbf{A}^\intercal \mathbf{A})^{-1} \mathbf{A}^\intercal 
        \left( \mathbb{E}_X[\mathbf{f}] - \hat{\mathbb{E}}_X [\mathbf{f}] \right) \right\|_2^2
        = O_p\!\left( \frac{L B(\kappa_A \kappa_X)^{2-\alpha}}{c_A n} \right).
    $$
\end{proof}

%
%
\begin{lemma}
    \label{lem:S(w*,f)_rate}
    Let $S(\mathbf{w}, f)$ be defined in \eqref{eq:S(w, f)}. Then, evaluated at the true weights $w^*$, we have
    $$
        S(w^*, f)
        =
        (\mathbf{A}^\intercal \mathbf{A})^{-1} \mathbf{A}^\intercal 
        \left( \mathbf{W}^* \mathbf{f} - \hat{\mathbb{E}}_X [\mathbf{f}] \right).
    $$
    Under the conditions of Lemma~\ref{lem:sigmaA} and Lemma~\ref{lem:M_nl}, it holds that
    $$
        \sup_{f \in \mathcal{H}(1)} S(w^*, f)
        =
        O_p\!\left( 
        \frac{L B (C_w + 1)^2 (\kappa_A \kappa_X)^{2 - \alpha} }{c_A n}
        \right).
    $$
\end{lemma}

\begin{proof}
    Define
    \begin{align*}
        \mathbf{S}_1(f) &= 
            (\mathbf{A}^\intercal \mathbf{A})^{-1} \mathbf{A}^\intercal \left( \mathbf{W}^* \mathbf{f} \right), \\
        \mathbf{S}_2(f) &= 
            (\mathbf{A}^\intercal \mathbf{A})^{-1} \mathbf{A}^\intercal \left( \hat{\mathbb{E}}_X [\mathbf{f}] \right).
    \end{align*}
    Then
    $$
        S(w^*, f) = \| \mathbf{S}_1(f) - \mathbf{S}_2(f) \|_2^2.
    $$
    Write
    $$
        \mathbf{S}_1(f) - \mathbf{S}_2(f)
        =
        \bigl( \mathbf{S}_1(f) - \boldsymbol{\mu}_1(f) \bigr)
        -
        \bigl( \mathbf{S}_2(f) - \boldsymbol{\mu}_2(f) \bigr)
        +
        \bigl( \boldsymbol{\mu}_1(f) - \boldsymbol{\mu}_2(f) \bigr),
    $$
    where $\boldsymbol{\mu}_k(f) = \mathbb{E}[\mathbf{S}_k(f)]$, $k=1,2$.

    Under correct specification of the weights, the bias terms vanish, i.e.,
    $
        \boldsymbol{\mu}_1(f) = \boldsymbol{\mu}_2(f)
    $
    for all $f \in \mathcal{H}(1)$. Hence,
    $$
        S(w^*, f)
        =
        \left\|
        \bigl( \mathbf{S}_1(f) - \boldsymbol{\mu}_1(f) \bigr)
        -
        \bigl( \mathbf{S}_2(f) - \boldsymbol{\mu}_2(f) \bigr)
        \right\|_2^2.
    $$
    Using $\|a-b\|_2^2 \le 2\|a\|_2^2 + 2\|b\|_2^2$, we obtain
    $$
        \sup_{f \in \mathcal{H}(1)} S(w^*, f)
        \le
        2 \sup_{f \in \mathcal{H}(1)} \left\| \mathbf{S}_1(f) - \boldsymbol{\mu}_1(f) \right\|_2^2
        +
        2 \sup_{f \in \mathcal{H}(1)} \left\| \mathbf{S}_2(f) - \boldsymbol{\mu}_2(f) \right\|_2^2.
    $$
    By Lemma~\ref{lem:S(w*,f)_part1},
    $$
        \sup_{f \in \mathcal{H}(1)} 
        \left\| \mathbf{S}_1(f) - \boldsymbol{\mu}_1(f) \right\|_2^2
        =
        O_p\!\left( 
        \frac{L B (C_w + 1)^2 (\kappa_A \kappa_X)^{2 - \alpha} }{c_A n}
        \right),
    $$
    and by Lemma~\ref{lem:S(w*,f)_part2},
    $$
        \sup_{f \in \mathcal{H}(1)} 
        \left\| \mathbf{S}_2(f) - \boldsymbol{\mu}_2(f) \right\|_2^2
        =
        O_p\!\left( 
        \frac{L B (\kappa_A \kappa_X)^{2 - \alpha} }{c_A n}
        \right).
    $$
    Therefore,
    $$
        \sup_{f \in \mathcal{H}(1)} S(w^*, f)
        =
        O_p\!\left( 
        \frac{L B (C_w + 1)^2 (\kappa_A \kappa_X)^{2 - \alpha} }{c_A n}
        \right).
    $$
\end{proof}

%
%
\begin{lemma} \label{lem:p(w*)}
    Let $p(w)$ be defined in \eqref{eq:p(w)}. Then, evaluated at the true weights $w^*$,
    $$
        p(w^*) 
        = 
        \left\| (\mathbf{A}^\intercal \mathbf{A})^{-1} \mathbf{A}^\intercal \mathbf{W}^* \right\|_F^2.
    $$
    Under Assumption~\ref{ass:C_w} and the conditions of Lemma~\ref{lem:sigmaA}, it holds that
    $$
        p(w^*) = O_p\!\left( \frac{L C_w^2}{c_A n} \right).
    $$
\end{lemma}

\begin{proof}
    Since $\mathbf{A}^\intercal \mathbf{A} = n \hat{\boldsymbol{\Sigma}}_A$, we have
    $
        (\mathbf{A}^\intercal \mathbf{A})^{-1} = \frac{1}{n} \hat{\boldsymbol{\Sigma}}_A^{-1}.
    $
    Hence,
    \begin{align*}
        p(w^*)
        &=
        \left\| (\mathbf{A}^\intercal \mathbf{A})^{-1} \mathbf{A}^\intercal \mathbf{W}^* \right\|_F^2 \\
        &=
        \frac{1}{n^2}
        \left\| \hat{\boldsymbol{\Sigma}}_A^{-1} \mathbf{A}^\intercal \mathbf{W}^* \right\|_F^2 \\
        &\le
        \frac{1}{n^2}
        \left\| \hat{\boldsymbol{\Sigma}}_A^{-1} \right\|_{op}^2
        \left\| \mathbf{A}^\intercal \mathbf{W}^* \right\|_F^2.
    \end{align*}

    The $(l,i)$-th entry of $\mathbf{A}^\intercal \mathbf{W}^*$ is $a_{il} w_i^*$. Thus,
    $$
        \left\| \mathbf{A}^\intercal \mathbf{W}^* \right\|_F^2
        =
        \sum_{l=1}^L \sum_{i=1}^n (a_{il} w_i^*)^2.
    $$
    Since $0 \le a_{il} \le 1$ and, by Assumption~\ref{ass:C_w}, $|w_i^*| \le C_w$, we have
    $$
        (a_{il} w_i^*)^2 \le C_w^2,
    $$
    which implies
    $$
        \left\| \mathbf{A}^\intercal \mathbf{W}^* \right\|_F^2 \le n L C_w^2.
    $$

    Therefore,
    $$
        p(w^*)
        \le
        \left\| \hat{\boldsymbol{\Sigma}}_A^{-1} \right\|_{op}^2
        \frac{L C_w^2}{n}.
    $$
    By Lemma~\ref{lem:sigmaA},
    $
        \left\| \hat{\boldsymbol{\Sigma}}_A^{-1} \right\|_{op}^2 = O_p(c_A^{-1}),
    $
    and hence
    $$
        p(w^*) = O_p\!\left( \frac{L C_w^2}{c_A n} \right).
    $$
\end{proof}

%
%
\begin{lemma}
    \label{lem:S(w_hat,f)_p(w_hat)_rate}
    Let $\hat{w}$ be defined in \eqref{eq:w_hat_def}. 
    Under Assumptions~\ref{ass:Sigma_A}, \ref{ass:C_w}, \ref{ass:kernels}, and \ref{ass:entropy}, it holds that
    \begin{align*}
        \sup_{f \in \mathcal{H}(1)} S(\hat{w}, f) &= O_p\!\left( \frac{L}{c_A n} \left( B (C_w + 1)^2 (\kappa_A \kappa_X)^{2 - \alpha} + \lambda C_w^2 \right) \right),
        \\
        p(\hat{w}) &= O_p\!\left( \frac{L}{c_A n} \left( \frac{ B (C_w + 1)^2 (\kappa_A \kappa_X)^{2 - \alpha}}{\lambda} + C_w^2 \right) \right).
    \end{align*}
\end{lemma}

\begin{proof}
    By definition of $\hat{w}$ as the minimizer of \eqref{eq:w_hat_def}, $\hat{w}$ satisfies
    $$
        \sup_{f \in \mathcal{H}(1)} S(\hat{w}, f) + \lambda p(\hat{w})
        \le
        \sup_{f \in \mathcal{H}(1)} S(w^*, f) + \lambda p(w^*).
    $$

    Since $p(\hat{w}) \ge 0$,
    $$
        \sup_{f \in \mathcal{H}(1)} S(\hat{w}, f)
        \le
        \sup_{f \in \mathcal{H}(1)} S(w^*, f) + \lambda p(w^*).
    $$

    By Lemmas~\ref{lem:S(w*,f)_rate} and \ref{lem:p(w*)},
    \begin{align*}
        \sup_{f \in \mathcal{H}(1)} S(w^*, f) &=
            O_p\!\left( \frac{L B (C_w + 1)^2 (\kappa_A \kappa_X)^{2 - \alpha}}{c_A n} \right), \\
        p(w^*) &= O_p\!\left( \frac{L C_w^2}{c_A n} \right).
    \end{align*}

    Hence,
    \begin{align*}
        \sup_{f \in \mathcal{H}(1)} S(\hat{w}, f) &= 
            O_p\!\left( \frac{L}{c_A n} \left( B (C_w + 1)^2 (\kappa_A \kappa_X)^{2 - \alpha} + \lambda C_w^2 \right) \right).
    \end{align*}

    Since $\sup_{f \in \mathcal{H}(1)} S(\hat{w}, f) \ge 0$, we also have
    $$
        \lambda p(\hat{w}) \le
        \sup_{f \in \mathcal{H}(1)} S(w^*, f) + \lambda p(w^*),
    $$
    which implies
    $$
        \lambda p(\hat{w}) =
            O_p\!\left( \frac{L}{c_A n} \left( B (C_w + 1)^2 (\kappa_A \kappa_X)^{2 - \alpha} + \lambda C_w^2 \right) \right).
    $$
    Dividing by $\lambda$ yields
    $$
        p(\hat{w}) = O_p\!\left( \frac{L}{c_A n} \left( \frac{B (C_w + 1)^2 (\kappa_A \kappa_X)^{2 - \alpha}}{\lambda} + C_w^2 \right) \right).
    $$
\end{proof}

%
%
\begin{lemma}
\label{lem:||m||_infty}
    Let $m \in \mathcal{H}$. Define 
    $$
        \| m \|_\infty := \sup_{(a,x) \in \mathcal{A} \times \mathcal{X}} | m(a, x) |.
    $$
    Under Assumption~\ref{ass:kernels}, it holds that
    $$
        \| m \|_\infty < \infty.
    $$
\end{lemma}

\begin{proof}
    Since $m \in \mathcal{H}$, the reproducing property implies that, for any $(a,x) \in \mathcal{A} \times \mathcal{X}$,
    $$
        m(a,x) = \langle m, K((a,x), \cdot) \rangle_{\mathcal{H}}.
    $$
    By the Cauchy--Schwarz inequality,
    $$
        |m(a,x)| 
        \le \|m\|_{\mathcal{H}} \, \|K((a,x), \cdot)\|_{\mathcal{H}}.
    $$
    Using the reproducing kernel property again,
    $$
        \|K((a,x), \cdot)\|_{\mathcal{H}}^2 = K((a,x),(a,x)).
    $$
    Under Assumption~\ref{ass:kernels},
    $$
        K((a,x),(a,x)) = K_A(a,a) K_X(x,x) \le \kappa_A^2 \kappa_X^2.
    $$
    Hence,
    $$
        |m(a,x)| \le \kappa_A \kappa_X \|m\|_{\mathcal{H}}
    $$
    for all $(a,x) \in \mathcal{A} \times \mathcal{X}$. Taking the supremum over $(a,x)$ yields
    $$
        \|m\|_\infty \le \kappa_A \kappa_X \|m\|_{\mathcal{H}} < \infty.
    $$
\end{proof}

%
%
\begin{lemma}
    \label{lem:m-beta*}
    Suppose Assumptions~\ref{ass:Sigma_A}, \ref{ass:C_w}, and \ref{ass:kernels} hold.
    Then,
    $$
        \left\|
        \hat{\Sigma}_A^{-1}
        \left(
        \frac{1}{n}\sum_{i=1}^n
        A_i\left( \mathbb{E}_X[m(A_i,X)] - A_i^\intercal \beta^* \right)
        \right)
        \right\|_2
        =
        O_p\!\left( \frac{\| m \|_\infty \sqrt{1 + c_A^{-2}}}{c_A \sqrt{n}} \right).
    $$
\end{lemma}

\begin{proof}
    Define
    $$
        Z_i := A_i\left( \mathbb{E}_X[m(A_i,X)] - A_i^\intercal \beta^* \right),
        \qquad i=1,\dots,n.
    $$
    Then
    $$
        \hat{\Sigma}_A^{-1}
        \left(
        \frac{1}{n}\sum_{i=1}^n
        A_i\left( \mathbb{E}_X[m(A_i,X)] - A_i^\intercal \beta^* \right)
        \right)
        =
        \hat{\Sigma}_A^{-1}
        \left(
        \frac{1}{n}\sum_{i=1}^n Z_i
        \right).
    $$

    We first show that \(Z_i\) is centered. By the definition of \(\beta^*\),
    \begin{align*}
        \mathbb{E}[Z_i]
        =
        \mathbb{E}[A_i \mathbb{E}_X[m(A_i, X)] - A_i A_i^\intercal \beta^* ] = 0.
    \end{align*}

    Next, we control the second moment of \(Z_i\). Since \(A_i\) lies on the simplex, \(\|A_i\|_2 \le 1\), and hence
    $$
        \|Z_i\|_2
        =
        \left\| A_i\left( \mathbb{E}_X[m(A_i,X)] - A_i^\intercal \beta^* \right) \right\|_2
        \le
        \left| \mathbb{E}_X[m(A_i,X)] - A_i^\intercal \beta^* \right|.
    $$
    Therefore,
    \begin{align*}
        \mathbb{E} \left[ \|Z_i\|_2^2 \right]
        &\le
        \mathbb{E}\left[ \left| \mathbb{E}_X[m(A_i,X)] - A^\intercal \beta^* \right|^2 \right] \\
        &\le
        2\mathbb{E}\left[ \mathbb{E}_X[m(A_i,X)]^2 \right] + 2\mathbb{E}\bigl[(A^\intercal \beta^*)^2\bigr].
    \end{align*}
    Since \(\|A\|_2 \le 1\),
    $$
        (A^\intercal \beta^*)^2 \le \|\beta^*\|_2^2,
    $$
    so
    $$
        \mathbb{E} \left[ \|Z_i\|_2^2 \right]
        \le
        2\mathbb{E}\left[ \mathbb{E}_X[m(A_i,X)]^2 \right] + 2 \| \beta^* \|^2_2.
    $$

    By Lemma~\ref{lem:||m||_infty}
    $$
        \mathbb{E}_X[m(A_i,X)]^2 \le \| m \|_\infty^2,
    $$
    and
    \begin{align*}
        \left\| \boldsymbol{\beta}^* \right\|_2 
        & \le 
        \left\| \boldsymbol{\Sigma}_A^{-1} \right\|_{op}  \left\| \mathbb{E}\left[ A \, \mathbb{E}_X[m(A,X)] \right] \right\|_2 \\
        & \le
        c_A^{-1} \| m \|_\infty
    \end{align*}

    Combining the above bounds yields
    $$
        \mathbb{E} \left[ \|Z_i\|_2^2 \right] 
        \le 
        2 \| m \|_\infty^2 + 2 c_A^{-2} \| m \|_\infty^2
        = 2 (1 + c_A^{-2}) \| m \|_\infty^2
    $$
    Therefore,
    $$
        \mathbb{E} \left[ \left\|
        \frac{1}{n}\sum_{i=1}^n Z_i
        \right\|_2^2 \right]
        =
        \frac{1}{n^2}\sum_{i=1}^n \mathbb{E} \left[ \|Z_i\|_2^2 \right]
        \le
        \frac{2(1 + c_A^{-2}) \| m \|_\infty^2}{n},
    $$
    since \(Z_1,\dots,Z_n\) are independent and centered. It follows that
    $$
        \left\|
        \frac{1}{n}\sum_{i=1}^n Z_i
        \right\|_2
        =
        O_p\!\left( \frac{\| m \|_\infty \sqrt{1 + c_A^{-2}}}{\sqrt{n}} \right).
    $$

    Finally, Lemma~\ref{lem:sigmaA} implies
    $$
        \left\| \hat{\Sigma}_A^{-1} \right\|_{op} = O_p\left( c_A^{-1} \right).
    $$
    Hence,
    $$
        \left\|
        \hat{\Sigma}_A^{-1}
        \left(
        \frac{1}{n}\sum_{i=1}^n Z_i
        \right)
        \right\|_2
        \le
        \|\hat{\Sigma}_A^{-1}\|_{op}
        \left\|
        \frac{1}{n}\sum_{i=1}^n Z_i
        \right\|_2
        =
        O_p\!\left( \frac{\| m \|_\infty \sqrt{1 + c_A^{-2}}}{c_A \sqrt{n}} \right).
    $$
    This completes the proof.

\end{proof}

%
%
\begin{lemma}
    \label{lem:cross_term_rate}
    Suppose Assumptions~\ref{ass:Sigma_A}, \ref{ass:C_w}, and \ref{ass:kernels} hold. Then,
    $$
        \left\|
        \hat{\boldsymbol{\Sigma}}_A^{-1}
        \left(
        \frac{1}{n}\sum_{i=1}^n
        A_i
        \left(
        \frac{1}{n}\sum_{j=1}^n m(A_i,X_j) - \mathbb{E}_X[m(A_i,X)]
        \right)
        \right)
        \right\|_2
        =
        O_p\!\left( \frac{\| m \|_\infty}{c_A\sqrt{n}} \right).
    $$
\end{lemma}

\begin{proof}
    Define
    $$
        T_n
        :=
        \frac{1}{n}\sum_{i=1}^n
        A_i
        \left(
        \frac{1}{n}\sum_{j=1}^n m(A_i,X_j)-\mathbb{E}_X[m(A_i,X)]
        \right).
    $$
    Let $Z_i=(A_i,X_i)$ and define
    $$
        h\bigl((a,x),(a',x')\bigr)
        :=
        a\Bigl(m(a,x')-\mathbb{E}_X[m(a,X)]\Bigr).
    $$
    Then
    $$
        T_n=\frac{1}{n^2}\sum_{i,j=1}^n h(Z_i,Z_j).
    $$

    Since
    $$
        \mathbb{E}\bigl[h(Z_1,Z_2)\bigr]
        =
        \mathbb{E}\left[
        A_1\left(m(A_1,X_2)-\mathbb{E}_X[m(A_1,X)]\right)
        \right]
        =0,
    $$
    the V-statistic is centered. Let
    $$
        h_1(z):=\mathbb{E}[h(z,Z_2)],
        \qquad
        h_2(z):=\mathbb{E}[h(Z_1,z)].
    $$
    For $z=(a,x)$,
    $$
        h_1(a,x)
        =
        a\,\mathbb{E}_{X'}\!\left[m(a,X')-\mathbb{E}_X[m(a,X)]\right]
        =0.
    $$
    Hence, by the Hoeffding decomposition,
    $$
        T_n
        =
        \frac{1}{n}\sum_{j=1}^n h_2(Z_j)+R_n,
    $$
    where $R_n$ is a degenerate remainder term satisfying
    $$
        \|R_n\|_2=O_p(n^{-1}).
    $$

    Moreover,
    $$
        h_2(a,x)
        =
        \mathbb{E}_{A'}\left[
        A' \left(m(A',x)-\mathbb{E}_X[m(A',X)]\right)
        \right].
    $$
    Since $A' \in \mathcal{A}$, we have $\|A'\|_2\le 1$. Also,
    $$
        \left|m(A',x)-\mathbb{E}_X[m(A',X)]\right|
        \le 2\|m\|_\infty.
    $$
    Therefore,
    $$
        \|h_2(a,x)\|_2
        \le
        \mathbb{E}_{A'}\left[
        \|A'\|_2 \left|m(A',x)-\mathbb{E}_X[m(A',X)]\right|
        \right]
        \le 2\|m\|_\infty.
    $$
    Thus, $\{h_2(Z_j)\}_{j=1}^n$ are i.i.d. mean zero random vectors with uniformly bounded second moments, which implies
    $$
        \left\|
        \frac{1}{n}\sum_{j=1}^n h_2(Z_j)
        \right\|_2
        =
        O_p\!\left( \frac{\|m\|_\infty}{\sqrt{n}} \right).
    $$
    Combining this with the bound for $R_n$ yields
    $$
        \|T_n\|_2
        =
        O_p\!\left( \frac{\|m\|_\infty}{\sqrt{n}} \right).
    $$

    Finally, by Lemma~\ref{lem:sigmaA},
    $$
        \left\| \hat{\boldsymbol{\Sigma}}_A^{-1} \right\|_{op}=O_p(c_A^{-1}),
    $$
    and hence
    $$
        \left\|
        \hat{\boldsymbol{\Sigma}}_A^{-1}T_n
        \right\|_2
        \le
        \left\| \hat{\boldsymbol{\Sigma}}_A^{-1} \right\|_{op}\,\|T_n\|_2
        =
        O_p\!\left( \frac{\|m\|_\infty}{c_A\sqrt{n}} \right).
    $$
\end{proof}

%
%
\begin{lemma} 
    \label{lem:error_term_bound}
    Under Assumption~\ref{ass:epsilon} and Lemma~\ref{lem:S(w_hat,f)_p(w_hat)_rate}, it holds that
    $$
        \left\| (\mathbf{A}^\intercal \mathbf{A})^{-1} \mathbf{A}^\intercal \hat{\mathbf{W}} \boldsymbol{\varepsilon} \right\|_2^2 
        =
        O_p\!\left( \frac{\sigma^2 L}{c_A n} \left( \frac{ B (C_w + 1)^2 (\kappa_A \kappa_X)^{2 - \alpha}}{\lambda} + C_w^2 \right) \right).
    $$
\end{lemma}

\begin{proof}
    Let $(\mathbf{A}^\intercal \mathbf{A})^{-1} \mathbf{A}^\intercal \hat{\mathbf{W}} = \mathbf{M}$. Conditional on $(\mathbf{A}, \mathbf{X})$, we have
     \begin{align*}
        \mathbb{E}\!\left[
        \left\| \mathbf{M} \boldsymbol{\varepsilon} \right\|_2^2
        \,\middle|\, \mathbf{A}, \mathbf{X}
        \right]
        &=
        \operatorname{tr}\!\left(
        \mathbf{M} \,
        \operatorname{Cov}(\boldsymbol{\varepsilon} \mid \mathbf{A}, \mathbf{X})
        \mathbf{M}^\intercal
        \right).
    \end{align*}
    By Assumption~\ref{ass:epsilon}, 
    $
    \operatorname{Cov}(\boldsymbol{\varepsilon} \mid \mathbf{A}, \mathbf{X})
    \preceq
    \sigma^2 \mathbf{I}_n,
    $
    hence
    \begin{align*}
        \mathbb{E}\!\left[
        \left\| \mathbf{M} \boldsymbol{\varepsilon} \right\|_2^2
        \,\middle|\, \mathbf{A}, \mathbf{X}
        \right]
        \le
        \sigma^2 \operatorname{tr}(\mathbf{M}\mathbf{M}^\intercal)
        =
        \sigma^2 \|\mathbf{M}\|_F^2.
    \end{align*}
    
    Noting that
    $
        \|\mathbf{M}\|_F^2 =
        \| (\mathbf{A}^\intercal \mathbf{A})^{-1} \mathbf{A}^\intercal \hat{\mathbf{W}} \|_F^2 = p(\hat{w}),
    $
    we obtain
    $$
        \mathbb{E}\!\left[
        \left\| \mathbf{M} \boldsymbol{\varepsilon} \right\|_2^2
        \,\middle|\, \mathbf{A}, \mathbf{X}
        \right]
        \le
        \sigma^2 p(\hat{w}).
    $$

    Taking expectations and applying Lemma~\ref{lem:S(w_hat,f)_p(w_hat)_rate} yields
    $$
        \left\| (\mathbf{A}^\intercal \mathbf{A})^{-1} \mathbf{A}^\intercal \hat{\mathbf{W}} \boldsymbol{\varepsilon} \right\|_2^2 
        =
        O_p\!\left( \frac{\sigma^2 L}{c_A n} \left( \frac{ B (C_w + 1)^2 (\kappa_A \kappa_X)^{2 - \alpha}}{\lambda} + C_w^2 \right) \right).
    $$

\end{proof}

%
%
\begin{lemma}
    \label{lem:normality}
    Let 
    $$
        T_n := \hat{\boldsymbol{\Sigma}}_A^{-1} \left[ \frac{1}{\sqrt{n}} \sum_{i=1}^n \hat{w}_i \varepsilon_i A_i \right],
    $$
    where $\hat{\boldsymbol{\Sigma}}_A = \frac{1}{n} \sum_{i=1}^n A_i A_i^\intercal$ and
    $$
        \mathbf{V}_n
        :=
        \hat{\boldsymbol{\Sigma}}_A^{-1} \left( \frac1n \sum_{i=1}^n \hat w_i^2 \sigma_i^2 A_i A_i^\intercal \right) \hat{\boldsymbol{\Sigma}}_A^{-1}.
    $$
    Under the conditions of Theorem~\ref{thm:awe_rate} and  Assumption~\ref{ass:eps_third_moment} and \ref{ass:weighted_stability}, it holds that
    $$
        \mathbf{V}_n^{-1/2} T_n \overset{d}{\longrightarrow} \mathcal{N}(\mathbf{0}, \mathbf{I} ).
    $$
\end{lemma}

\begin{proof}
    Define
    $$
        U_i :=
        \frac{1}{\sqrt{n}}
        \hat{\boldsymbol{\Sigma}}_A^{-1}
        \hat{w}_i A_i \varepsilon_i.
    $$
    Then $ T_n  = \sum_{i=1}^n U_i $.    
    Conditional on $\mathcal{D}_n := \{(A_1, X_1), \dots , (A_n, X_n) \}$, the variables $\{U_i\}_{i=1}^n$ are independent with
    $$ \mathbb{E} [U_i \mid \mathcal{D}_n] = 0. $$
    
    Their conditional covariance matrix is
    \begin{align*}
        \sum_{i=1}^n \mathrm{Cov}(U_i \mid \mathcal{D}_n)
        =
        \hat{\boldsymbol{\Sigma}}_A^{-1}
        \left(
            \frac{1}{n}\sum_{i=1}^n \hat{w}_i^2 \sigma_i^2 A_i A_i^\intercal
        \right)
        \hat{\boldsymbol{\Sigma}}_A^{-1}
        = \mathbf{V}_n.
    \end{align*}

    We verify the conditional Lyapunov condition with $\delta = 1$.
    Since $ \|A_i\|_2 \le 1 $ for all $i$ on the simplex,
    \begin{align*}
        \sum_{i=1}^n \mathbb{E} \left[ \| U_i \|^3_2 \mid \mathcal{D}_n \right]
        \le 
        n^{-3/2} \| \hat{\boldsymbol{\Sigma}}_A^{-1}\|_{op}^3 
        \sum_{i=1}^n |\hat{w}_i|^3 \|A_i\|^3_2 \mathbb{E} \left[ |\varepsilon_i|^3 \mid A_i, X_i \right].
    \end{align*}

    Under Assumptions~\ref{ass:Sigma_A},
    $$ \| \hat{\boldsymbol{\Sigma}}_A^{-1}\|_{op} = O_p(c_A^{-1}) = O_p(1). $$

    By Assumption~\ref{ass:eps_third_moment}
    $$
        \mathbb{E} \left[ |\varepsilon_i|^3 \mid A_i, X_i \right] \le C_{\varepsilon}.
    $$

    Furthermore,
    $$ \sum_{i=1}^n |\hat w_i|^3 \le n \|\hat w\|_\infty^3. $$

    Therefore,
    \begin{align*}
        \sum_{i=1}^n
        \mathbb E
        \left[
            \|U_i\|_2^3
            \mid
            \mathcal D_n
        \right]
        &=
        O_p(n^{-1/2})
        \|\hat w\|_\infty^3
        \\
        &=
        o_p(1),
    \end{align*}
    by Assumption~\ref{ass:weighted_stability}. Hence the conditional Lyapunov condition holds. Conditional on $\mathcal D_n$, the result follows from the conditional multivariate Lyapunov central limit theorem,
    $$ T_n \xrightarrow{d} \mathcal N(\mathbf 0, \mathbf{V}_n). $$

    Since the variance $\mathbf{V}_n$ has the inverse matrix by Assumption~\ref{ass:weighted_stability}, $\left( \mathbf{V}_n \right)^{-1/2}$ can exist. Therefore,
    $$
        \mathbf{V}_n^{-1/2} T_n \overset{d}{\longrightarrow} \mathcal{N}(\mathbf{0}, \mathbf{I} ).
    $$
\end{proof}

\end{document}